\documentclass{article}

\usepackage{amsmath,amssymb,amsthm,graphicx,proof,color}
\usepackage[UKenglish]{babel}

\usepackage{hyperref}
\usepackage[all]{xy}

\newcommand{\bis}{\mathrel{\mathchoice%
{\raisebox{.3ex}{$\,
  \underline{\makebox[.7em]{$\leftrightarrow$}}\,$}}%
{\raisebox{.3ex}{$\,
  \underline{\makebox[.7em]{$\leftrightarrow$}}\,$}}%
{\raisebox{.2ex}{$\,
  \underline{\makebox[.5em]{\scriptsize$\leftrightarrow$}}\,$}}%
{\raisebox{.2ex}{$\,
  \underline{\makebox[.5em]{\scriptsize$\leftrightarrow$}}\,$}}}}

\newcommand{\Kbis}{\bis_{\Box}}
\newcommand{\kwbis}{\bis_{\circ}}

\newcommand{\equkw}{\ensuremath{\equiv_{\circ}}}
\newcommand{\equk}{\ensuremath{\equiv_{\Box}}}

\hypersetup{
   colorlinks,%
   citecolor=blue,%
   filecolor=black,%
   linkcolor=red,%
   urlcolor=black
 }

\newcommand{\M}{\mathcal{M}}
\newcommand{\N}{\mathcal{N}}

\newcommand{\KwTop}{\ensuremath{\circ\top}}
\newcommand{\EquiKw}{\ensuremath{\circ\neg}}
\newcommand{\KwCon}{\ensuremath{\circ\land}}
\newcommand{\R}{\ensuremath{\texttt{R}}}
\newcommand{\TAUT}{\ensuremath{\texttt{TAUT}}}
\newcommand{\SUB}{\ensuremath{\texttt{US}}}
\newcommand{\MP}{\ensuremath{\texttt{MP}}}
\newcommand{\REKw}{\ensuremath{\texttt{RE}\circ}}

\newcommand{\KwTr}{\ensuremath{\circ{\bf 4}}}
\newcommand{\KwB}{\ensuremath{\circ{\bf B}}}
\newcommand{\KwEuc}{\ensuremath{\circ{\bf5}}}

\newcommand{\BP}{\ensuremath{\textbf{P}}}

\newcommand{\ML}{\ensuremath{\mathcal{L}(\Box)}}

\newcommand{\SNCL}{\ensuremath{\mathcal{L}(\blacktriangle)}}

\newcommand{\LEA}{\ensuremath{\mathcal{L}(\circ)}}

\newcommand{\SLCL}{\ensuremath{\mathbf{K}^\circ}}
\newcommand{\SLEA}{\ensuremath{\mathbf{K}^\circ}}
\newcommand{\SLCLTr}{\ensuremath{\mathbf{K4}^\circ}}
\newcommand{\SLCLB}{\ensuremath{\mathbf{KB}^\circ}}

\newcommand{\SLCLBEuc}{\ensuremath{\mathbf{KB5}^\circ}}
\newcommand{\lr}[1]{\langle #1 \rangle}
\newcommand{\lra}{\leftrightarrow}
\newcommand{\Lra}{\Leftrightarrow}

\newcommand{\jieproof}[1]{{\noindent {\bf Proof of Proposition~\ref{#1}:} \\}}

\renewcommand{\phi}{\varphi}

\newtheorem{theorem}{Theorem}
\newtheorem{lemma}[theorem]{Lemma}
\newtheorem{definition}[theorem]{Definition}
\newtheorem{proposition}[theorem]{Proposition}

\newtheorem{example}[theorem]{Example}

\newcommand{\weg}[1]{}

\title{Logics of Essence and Accident}
\author{Jie Fan}
\date{}

\begin{document}
\maketitle

\weg{\begin{abstract}
Usually, non-contingency is interpreted as `necessary truth or necessary falsity'; otherwise, it is contingent. In this paper, we consider a operator called {\em strong non-contingency}, such that a proposition is called strongly non-contingent, if no matter whether it is true or false, it does it necessarily; otherwise, it is weakly contingent. The strong version of non-contingency is consistent with many linguistic interpretations. we investigate a logic of strong non-contingency operator, and present axiomatizations for this logic over various frame classes. We also compare the relative expressivity of the logic of strong non-contingency and the usual non-contingency logic, and also the logic of essence and accident.
\end{abstract}}

\begin{abstract}
In the literature, {\em essence} is formalized in two different ways, either {\em de dicto}, or {\em de re}. Following~\cite{Marcos:2005}, we adopt its {\em de dicto} formalization: a formula is essential, if once it is true, it is necessarily true; otherwise, it is accidental. In this article, we study the model theory and axiomatization of the logic of essence and accident, i.e. the logic with essence operator (or accident operator) as the only primitive modality. We show that the logic of essence and accident is less expressive than modal logic on non-reflexive models, but the two logics are equally expressive on reflexive models. We prove that some frame properties are undefinable in the logic of essence and accident, while some are. We propose the suitable bisimulation for this logic, based on which we characterize the expressive power of this logic within modal logic and within first-order logic. We axiomatize this logic over various frame classes, among which the symmetric case is missing, and our method is more suitable than those in the literature. We also find a method to compute certain axioms used to axiomatize this logic over special frames in the literature. As a side effect, we answer some open questions raised in \cite{Marcos:2005}.
\end{abstract}


\section{Introduction}

As far back as Aristotle, like the notions of necessity, possibility and contingency, the notion of {\em essence} can also be related either to propositions ({\em de dicto}) or to objects ({\em de re}). The importance of the notion of essence is argued in \cite{essenceandmodality:1994}. 

\weg{The first person to formalize essence in terms of {\em de re} seems to be Kit Fine.}
The formalization of essence in terms of {\em de re} at least tracks back to Kit Fine. In his writing \cite{Fine:essence}, a logic of essence is proposed, where formulas of the form $\Box_F\phi$ express that $\phi$ is true in virtue of the essence of objects which $F$. A Hilbert-style quantified system E5 is given but without a semantics. \weg{For Fine, we should view metaphysical truth as a special case of essence, but not vise versa. Essence cannot be explained in terms of necessity
\cite[page 241]{Fine:essence}. In \cite{Fine:essence}, a logic of essence is proposed, and a Hilbert-style quantified system E5 is given but without a semantics. }In \cite{Fine:semantics}, a possible worlds semantics is presented, and a variant of E5 is shown to be sound and complete for the semantics. In \cite{Correia:essence}, a propositional version of E5 is established in accompany with an appropriate semantics, and it is shown that the system is sound and complete with respect to the proposed semantics. A new semantics for logics of essence is proposed in \cite{Girodani:2014}. \weg{The practices of Fine \cite{Fine:essence,Fine:semantics} does not formalize the notion of essence itself, or `what is essence'.}

\weg{As argued in \cite[pp.~3-4]{essenceandmodality:1994}, essence cannot be defined in the following way: an object essentially has a property if and only if it is necessary that the object has the property\weg{ (if it exists)}, because even if for example it is necessary that Socrates belongs to singleton Socrates\weg{ if Socrates and the singleton both exist}, we cannot say that Socrates essentially belongs to singleton Socrates because no part of the essence of Socrates belongs to the singleton. But that does not mean that essence cannot be defined as `truth implies necessary truth'. Instead, it seems that Fine agreed on the definition, since he said ``For a statement of essence is a statement of necessity and so it will, like any statement of necessity, be necessarily true if it is true at all.'', c.f.~\cite[p.~5]{essenceandmodality:1994}.}

There are also researchers who formalize essence in terms of {\em de dicto}.
In \cite{Small:2001}, in reconstructing G\"odel's ontological argument, {\em accidental truth}, i.e. {\em accident} is formalized as $\phi\land\neg\Box\phi$, i.e. true but not necessarily true. Accordingly, as the negation of accident, essence is formalized as $\phi\to\Box\phi$. It is said that the discussions of essential and accidental propositions at least tracks back to the XIX Century, see
\cite[p.~53]{Marcos:2005}. A logic of essence and accident is introduced in which essence is treated in the metaphysical usage in \cite{Marcos:2005}, where a complete axiomatization for the logic is shown with respect to the class of all frames. A simple axiomatization for arbitrary frames and its extensions over various frame classes are proposed in \cite{Steinsvold:2008}, but the case for symmetric frames is missing. Even though the completeness proofs thereof are simple, his method has a defect: on one hand, the canonical relation, thus the canonical frame, is automatically provided to be reflexive; on the other hand, the underlying semantics is defined on arbitrary frames, rather than on reflexive frames\weg{, and it is also shown in \cite[Prop.~3.5]{Steinsvold:2008} that $\SLEA$ is sound and complete with respect to the class of all frames}. This means that there is a non-correspondence between syntax and semantics in the logic of essence and accident. Oblivious to the literature on the logic of essence and accident, in \cite{steinsvold:2008b} the author provides a topological semantics for {\em a logic of unknown truths} and shows its completeness over the class of $\mathcal{S}4$ models.

The accident operator has various meanings in different contexts. For instance, in the setting of provability logic, `accident' means `true but unprovable', thus `$\phi$ is accident' means `$\phi$ is a G\"odel sentence' \cite{Kushida:2010}; in the setting of epistemic logic, `accident' means `unknown truths', thus `$\phi$ is accident' means `$\phi$ is true but unknown to the agent' \cite{steinsvold:2008b}.

In this article, we will follow the formalization of essence in~\cite{Marcos:2005}, study the notions of essence and accident from viewpoint of {\em de dicto}. We will discuss the model theory of the logic of essence and accident, propose some axiomatizations, whose completeness are shown with a more suitable method than those in the literature, and give an automatic method to compute certain axioms needed to characterize this logic over special frames.

The paper is organized as follows. Section~\ref{sec.langsem} introduces the language of the logic of essence and accident. Section~\ref{sec.exp} compares the relative expressive power of the logic of essence and accident and modal logic. Section~\ref{sec.framecor} explores the frame definability. We propose the bisimulation notion suitable for the logic of essence and accident in Section~\ref{sec.bis},\weg{\footnote{The idea of this part is announced without proofs in \cite{Fan:2015}.}} based on which we characterize the expressive power of this logic within modal logic and within first-order logic in Section~\ref{sec.char}. Section~\ref{sec.axiomatizations} axiomatizes the logic of essence and accident over various frames. In Section \ref{sec.comparison}, we compare our work with the literature on the logic of essence and accident and the modal logic of G\"odel sentence. We conclude with some future work in Section~\ref{sec.conclusion}.

\section{Language and Semantics} \label{sec.langsem}

First, we introduce the following language with essence operator and necessity operator as modalities, although we will focus on the language of logic of essence and accident.

\begin{definition}[Logical language $\mathcal{L}(\circ,\Box)$] Let $\BP$ be a set of propositional variables, the logical language $\mathcal{L}(\circ,\Box)$ is defined as follows:
$$\phi::=p\mid \neg\phi\mid(\phi\land\phi)\mid\circ\phi\mid\Box\phi$$
where $p\in\BP$.
Without the construct $\circ\phi$, we obtain the {\em language of modal logic $\ML$}; without the construct $\Box\phi$, we obtain the {\em language of essence and accident $\LEA$}. If $\phi\in\mathcal{L}(\circ,\Box)$, we say $\phi$ is an $\mathcal{L}(\circ,\Box)$-formula; if $\phi\in\ML$, we say $\phi$ is an $\ML$-formula; if $\phi\in\LEA$, we say $\phi$ is an $\LEA$-formula.
\end{definition}

Intuitively, $\circ\phi$ is read `it is essential that $\phi$', and $\Box\phi$ is read `it is necessary that $\phi$'. Other operators are defined as usual; in particular, $\bullet\phi$ is defined as $\neg\circ\phi$, read `it is accidental that $\phi$'. Note that $\bullet$ is {\em not} the dual of $\circ$.

\begin{definition}[Model] A {\em frame} is a tuple $\mathcal{F}=\lr{S,R}$, where $S$ is a nonempty set of possible worlds, $R$ is a binary relation over $S$. A {\em model} is a tuple $\M=\lr{\mathcal{F},V}$, where $V$ is a valuation function from $\BP$ to $\mathcal{P}(S)$. A {\em pointed model} $(\M,s)$ is a model $\M$ with a designated world $s$ in $\M$. We always omit the parentheses around $(\M,s)$ whenever convenient. We sometimes write $s\in\M$ for $s\in S$. We write $R(s)=\{t\in S\mid sRt\}$. We write $F_{\mathcal{T}}$ for the class of reflexive frames.
\end{definition}

\weg{\begin{definition}[Model] A {\em model} is a triple $\M=\lr{S,R,V}$, where $S$ is a nonempty set of possible worlds, $R$ is a binary relation over $S$, $V$ is a valuation function from $\BP$ to $\mathcal{P}(S)$. A {\em pointed model} $(\M,s)$ is a model $\M$ with a designated world $s$ in $\M$. We always omit the parentheses around $(\M,s)$ whenever convenient. We sometimes write $s\in\M$ for $s\in S$. We write $R(s)=\{t\in S\mid sRt\}$. A {\em frame} is a model without the valuation. We write $F_{\mathcal{T}}$ for the class of reflexive frames.
\end{definition}}

\begin{definition}[Semantics]\label{def.semantics} Given a pointed model $(\M,s)$ and an $\mathcal{L}(\circ,\Box)$-formula $\phi$, the satisfaction relation $\vDash$ is defined as follows:\footnote{We here use the notation $\&,\forall,\Rightarrow,\Leftrightarrow$, respectively, to stand for the metalanguage `and', `for all', `if $\cdots$ then $\cdots$', `if and only if'.}
\[
\begin{array}{|lcl|}
\hline
\M,s\vDash p&\Lra& s\in V(p)\\
\M,s\vDash\neg\phi&\Lra&\M,s\nvDash\phi\\
\M,s\vDash\phi\land\psi&\Lra&\M,s\vDash\phi~\&~\M,s\vDash\psi\\
\M,s\vDash\circ\phi&\Lra&(\mathcal{M},s\vDash\phi\Rightarrow\forall t(sRt\Rightarrow\M, t\vDash\phi))\\
\M,s\vDash\Box\phi&\Lra&\forall t(sRt\Rightarrow\M, t\vDash\phi)\\
\hline
\end{array}
\]
If $\M,s\vDash\phi$, we say $\phi$ is {\em true}, or {\em satisfied} at $(\M,s)$, sometimes we write $s\vDash\phi$; if for all $s\in \M$ we have $\M,s\vDash\phi$, we say $\phi$ is {\em valid on $\M$} and write $\M\vDash\phi$; if for all $\M$ based on $\mathcal{F}$ we have $\M\vDash\phi$, we say $\phi$ is {\em valid on $\mathcal{F}$} and write $\mathcal{F}\vDash\phi$; if for all $\mathcal{F}$ in a class of frames $F$ we have $\mathcal{F}\vDash\phi$, we say $\phi$ is {\em valid on $F$} and write $F\vDash\phi$; if the class of frames $F$ in question is arbitrary, then we say $\phi$ is {\em valid} and write $\vDash\phi$. We say $\phi$ is {\em satisfiable}, if $\nvDash\neg\phi$. The case for a set of formula is similarly defined. Given any two pointed models $(\M,s)$ and $(\N,t)$, if they satisfy the same $\LEA$-formulas, we say they are {\em $\circ$-equivalent}, notation: $(\M,s)\equkw(\N,t)$; if they satisfy the same $\ML$-formulas, we say they are {\em $\Box$-equivalent}, notation: $(\M,s)\equk(\N,t)$.
\end{definition}
\weg{\begin{proposition}
$\vDash \phi\to(\Box \phi\lra\circ \phi)$, $\vDash \neg \phi\to(\Diamond\phi\lra\neg\circ\neg \phi)$.
\end{proposition}}

Under the semantics, it is not hard to show that
\begin{proposition}\label{prop.equivalent} Let $\phi\in\mathcal{L}(\circ,\Box)$. Then
$F_{\mathcal{T}}\vDash \Box \phi\lra \phi\land \circ \phi$ and $F_{\mathcal{T}}\vDash\Diamond\phi\lra(\phi\vee\neg\circ\neg\phi)$.
\end{proposition}

Proposition~\ref{prop.equivalent} is very important. It guides us to find the desired axioms for characterizing $\LEA$ over certain frame classes, as we will see in Section \ref{sec.axiomatization}.

\section{Expressivity} \label{sec.exp}

In this section, we compare the relative expressivity of $\LEA$ and $\ML$. A related technical definition is introduced as follows.

\begin{definition}[Expressivity]
Let logical languages $\mathcal{L}_1$ and $\mathcal{L}_2$ be interpreted on the same class $M$ of models,
\begin{itemize}
\item $\mathcal{L}_2$ is {\em at least as expressive as} $\mathcal{L}_1$, notation: $\mathcal{L}_1\preceq \mathcal{L}_1$, if for any $\phi\in \mathcal{L}_1$, there exists $\psi\in\mathcal{L}_2$ such that for all $(\M,s)\in M$, we have $\M,s\vDash\phi\leftrightarrow\psi$.
\item $\mathcal{L}_1$ and $\mathcal{L}_2$ are {\em equally expressive}, notation: $\mathcal{L}_1\equiv\mathcal{L}_2$, if $\mathcal{L}_1\preceq \mathcal{L}_2$ and $\mathcal{L}_2\preceq\mathcal{L}_1$.
\item $\mathcal{L}_1$ is {\em less expressive than} $\mathcal{L}_2$, or $\mathcal{L}_2$ is {\em more expressive than} $\mathcal{L}_1$, notation: $\mathcal{L}_1\prec\mathcal{L}_2$, if $\mathcal{L}_1\preceq \mathcal{L}_2$ and $\mathcal{L}_2\not\preceq\mathcal{L}_1$.
\end{itemize}
\end{definition}

\begin{proposition}\label{prop.lessexp}
$\LEA$ is less expressive than \ML\ on the class of $\mathcal{K}$-models, $\mathcal{B}$-models, $4$-models, $5$-models.
\end{proposition}

\begin{proof}
Define a translation $t$ from $\LEA$ to \ML:
\[
\begin{array}{lll}
t(\top)&=&\top\\
t(p)&=&p\\
t(\neg\phi)&=&\neg t(\phi)\\
t(\phi\land\psi)&=&t(\phi)\land t(\psi)\\
t(\circ\phi)&=&t(\phi)\to\Box t(\phi)\\
\end{array}
\]

It is clear that $t$ is a truth-preserving translation. Therefore \ML\ is at least as expressive as \LEA.

Now consider the following pointed models $(\mathcal{M},s)$ and $(\mathcal{N},t)$, which can be distinguished by an \ML-formula $\Box\bot$, but cannot be distinguished by any $\LEA$-formulas:

\medskip

$$
\xymatrix{\mathcal{M}:\ \ \ s:p\ar@(ur,ul) & & & \mathcal{N}:\ \ \ t:p}
$$

It is easy to check $\mathcal{M}$ and $\mathcal{N}$ are both symmetric, transitive, and Euclidean. By induction we prove that for all $\phi\in\LEA$, $\mathcal{M},s\vDash\phi$ iff $\mathcal{N},t\vDash\phi$. The base cases and boolean cases are straightforward. For the case of $\circ\phi$, it is not hard to show that $\mathcal{M},s\vDash\circ\phi$ and $\mathcal{N},t\vDash\circ\phi$ (note that here we do not need to use the induction hypothesis), thus $\mathcal{M},s\vDash\circ\phi$ iff $\mathcal{N},t\vDash\circ\phi$, as desired.
\end{proof}

As for the case of $\mathcal{D}$-models, the result about the relative expressivity of $\LEA$ and $\ML$ is same as previous, but the proof is much more sophisticated, which needs {\em simultaneous induction}.

\begin{proposition}\label{prop.lessexp-d}
\LEA\ is less expressive than $\ML$ on the class of $\mathcal{D}$-models.
\end{proposition}

\begin{proof}
By the translation $t$ in the proof of Proposition \ref{prop.lessexp}, we have $\LEA\preceq\ML$.

Consider the following pointed models $(\M,s)$ and $(\N,s')$, which can be distinguished by an $\ML$-formula $\Box\Box p$, but cannot be distinguished by any $\LEA$-formulas:
$$
\xymatrix{\mathcal{M}:\ \ \ s:p\ar[rr]&&t:\neg p\ar@(ur,ul)\ar[ll] & & \mathcal{N}:\ \ \ s':p\ar[rr]&& t':\neg p\ar[ll]}
$$

It is not hard to see that $\M$ and $\N$ are both serial. By induction on $\phi\in\LEA$, we show simultaneously that for all $\phi$, (i) $\M,s\vDash\phi$ iff $\N,s'\vDash\phi$, and (ii) $\M,t\vDash\phi$ iff $\N,t'\vDash\phi$. The nontrivial case is $\circ\phi$.

For (i), we have the following equivalences:
\[
\begin{array}{lll}
\M,s\vDash\circ\phi&\stackrel{\text{semantics}}\Longleftrightarrow& s\vDash\phi\text{ implies }t\vDash\phi\\
&\stackrel{\text{IH for (i)}}\Longleftrightarrow& s'\vDash\phi\text{ implies }t\vDash\phi\\
&\stackrel{\text{(ii)}}\Longleftrightarrow&s'\vDash\phi\text{ implies }t'\vDash\phi\\
&\stackrel{\text{semantics}}\Longleftrightarrow& \N,s'\vDash\circ\phi\\
\\
\M,t\vDash\circ\phi&\stackrel{\text{semantics}}\Longleftrightarrow& t\vDash\phi\text{ implies }(s\vDash\phi\text{ and }t\vDash\phi)\\
&\Longleftrightarrow& t\vDash\phi\text{ implies }s\vDash\phi\\
&\stackrel{\text{IH for (ii)}}\Longleftrightarrow& t'\vDash\phi\text{ implies }s\vDash\phi\\
&\stackrel{\text{(i)}}\Longleftrightarrow&t'\vDash\phi\text{ implies }s'\vDash\phi\\
&\stackrel{\text{semantics}}\Longleftrightarrow& \N,t'\vDash\circ\phi\\
\end{array}\]

Therefore, $(\M,s)$ and $(\N,s')$ cannot be distinguished by any $\LEA$-formulas.
\end{proof}

However, on the $\mathcal{T}$-models, the situation is different.
\begin{proposition}\label{prop.equallyexp}
\LEA\ and \ML\ are equally expressive on the class of $\mathcal{T}$-models.
\end{proposition}

\begin{proof}
By the translation $t$ in the proof of Proposition~\ref{prop.lessexp}, we have $\LEA\preceq\ML$. Now define another translation $t'$ from \ML\ to $\LEA$, where the base cases and Boolean cases are similar to the corresponding cases for $t$, and $t'(\Box\phi)=\circ t'(\phi) \land t'(\phi)$. It is straightforward to show that $t'$ is a truth-preserving translation, due to Prop.~\ref{prop.equivalent}. Thus $\ML\preceq\LEA$, and therefore $\LEA\equiv\ML$.
\end{proof}

\section{Frame correspondence} \label{sec.framecor}

In \cite[Corollary~4.3]{Marcos:2005}, the five basic frame properties, except for symmetry, are shown by using the method of mirror reduction, to be undefinable in $\LEA$. As an open question (Open 4.4 there), the author would like to know which frame properties are definable in $\LEA$. This question is answered partly in \cite{Steinsvold:2008}, where the following results are established.
\weg{In \cite{Steinsvold:2008},  which are definable in $\LEA$, thus partly answered the open question raised in \cite[Open 4.4]{Marcos:2005}.}
\begin{proposition}\label{prop.already}\cite[Prop.~2.2, Prop.~2.5, Prop.~3.8]{Steinsvold:2008}
\begin{enumerate}
\item \cite[Prop.~2.2 without proof]{Steinsvold:2008}\label{prop.one}
The property of weak transitivity, viz. $\forall x\forall y\forall z(xRy\land yRz\land x\neq z\to xRz)$, is defined by $\circ p\land p\to\circ(\circ p\land p)$. Thus the property of weak transitivity is definable in $\LEA$.
\item \cite[Prop.~2.5]{Steinsvold:2008}
The property of weak connectedness, viz. $\forall x\forall y\forall z(xRy\land xRz\to yRz\lor y=z\lor zRy)$, is defined by $\circ(\circ p\land p\to q)\vee\circ(\circ q\land q\to p)$. Thus the property of weak connectedness is definable in $\LEA$.
\item \cite[Prop.~3.8 without proof]{Steinsvold:2008}\label{prop.two}
The property of weak-weak-Euclidicity, viz. $\forall x\forall y\forall z(xRy\land xRz\land x\neq z\land y\neq z\to yRz)$, is defined by $\neg\circ\neg p\to\circ(\circ\neg p\to p)$. Thus the property of weak-weak-Euclidicity is definable in $\LEA$.
\end{enumerate}
\end{proposition}

In this section, we first answer affirmatively the case for symmetry, thus completing the spectrum of cases for the five basic frame properties. Apart from this, we will also give other results.

\begin{proposition}\label{prop.definable-b}
The frame property of symmetry is definable in $\LEA$.
\end{proposition}

\begin{proof}
Given any frame $\mathcal{F}=\lr{S,R}$. We show that
$$\mathcal{F}\vDash\forall x\forall y(xRy\to yRx)\text{ iff }\mathcal{F}\vDash p\to\circ(\circ\neg p\to p).$$

Left-to-right: Suppose that $\mathcal{F}$ is symmetric, to show $\mathcal{F}\vDash p\to\circ(\circ\neg p\to p)$. For this, let $\M$ be an arbitrary model based on $\mathcal{F}$ and any $s\in S$. Assume that that $\M,s\vDash p$ (thus $s\vDash \circ\neg p\to p$) and $t$ is a successor of $s$ such that $t\vDash \circ \neg p$. As $R$ is symmetric, $tRs$. If $t\vDash \neg p$, then by the semantics of $\circ$, we should have $s\vDash\neg p$, contrary to the assumption. Then $t\vDash p$, and thus $t\vDash \circ \neg p\to p$. Since $t$ is arbitrary, we have $s\vDash\circ(\circ \neg p\to p)$. Therefore $s\vDash p\to\circ(\circ\neg p\to p)$, as desired.

Right-to-left: Suppose that $\mathcal{F}$ is not symmetric, to show $\mathcal{F}\nvDash p\to\circ(\circ\neg p\to p)$. By assumption, there exist $s,t\in S$ such that $sRt$ but {\em not} $tRs$, thus $s\neq t$. Define a valuation $V$ on $\mathcal{F}$ as $V(p)=\{s\}$. Obviously, $\lr{\mathcal{F},V},s\vDash p$, thus $s\vDash \circ\neg p\to p$. Furthermore, $t\vDash\neg p$, and given any $u$ such that $tRu$, we have $u\neq s$, thus $u\vDash\neg p$, hence $t\vDash\circ\neg p\land\neg p$, viz. $t\nvDash\circ\neg p\to p$. From this and $s\vDash \circ\neg p\to p$, it follows that $s\nvDash\circ(\circ\neg p\to p)$, then $\lr{\mathcal{F},V},s\nvDash p\to\circ(\circ\neg p\to p)$. We now conclude that $\mathcal{F}\nvDash p\to\circ(\circ\neg p\to p)$.
\end{proof}

\begin{proposition}
The frame property of coreflexivity, viz. $\forall x\forall y(xRy\to x=y)$, is defined by $\circ p$. Thus coreflexivity is definable in $\LEA$.
\end{proposition}

\begin{proof}
Let $\mathcal{F}=\langle S,R\rangle$. We will show that $\mathcal{F}\vDash\forall x\forall y(xRy\to x=y)$ iff $\mathcal{F}\vDash\circ p$.

Suppose that $\mathcal{F}\vDash\forall x\forall y(xRy\to x=y)$. Then given any $\M$ based on $\mathcal{F}$ and any $s\in S$, if for each $t$ with $sRt$, we have $s=t$, then $s\vDash p$ implies $t\vDash p$, and thus $s\vDash \circ p$. Therefore $\mathcal{F}\vDash\circ p$.

Conversely, suppose that $\mathcal{F}\nvDash\forall x\forall y(xRy\to x=y)$. Then there are $s,t\in S$ such that $sRt$ but $s\neq t$. Define a valuation $V$ on $\mathcal{F}$ such that $V(p)=\{s\}$, then $s\vDash p$ but $t\nvDash p$, and thus $\langle\mathcal{F},V\rangle,s\nvDash\circ p$. Therefore $\mathcal{F}\nvDash\circ p$.
\end{proof}

The following result is an equivalent but different form of Proposition~\ref{prop.already}, item \ref{prop.one} and item \ref{prop.two}, respectively. For the proof details we refer to Appendix \ref{appendix.proofs}.
\begin{proposition}\label{prop.new}\
\begin{enumerate}
\item \label{prop.newone}
The frame property $\forall x\forall y\forall z(xRy\land yRz\land x\neq y\land y\neq z\land x\neq z\to xRz)$ is defined by $\circ p\land p\to\circ(\circ p\land p)$.
\item \label{prop.newtwo}
The frame property $\forall x\forall y\forall z(xRy\land xRz\land x\neq y\land x\neq z\land y\neq z\to yRz)$ is defined by $\neg\circ\neg p\to\circ(\circ\neg p\to p)$.
\end{enumerate}
\end{proposition}

\section{Bisimulation}\label{sec.bis}

In this section, we propose the suitable notion of bisimulation for the logic of essence and accident $\LEA$. It is announced but without definitions or proofs in~\cite{Fan:2015} that the bisimulation for $\LEA$ is similar to that for the logic of strong noncontingency $\SNCL$.

We first recall the standard notion of bisimulation for modal logic $\ML$.

\begin{definition}[$\Box$-bisimulation]
Let $\M=\langle S,R,V\rangle$ and $\M'=\lr{S',R',V'}$. A nonempty binary relation $Z$ is called a {\em $\Box$-bisimulation} between $\M$ and $\M'$, if $sZs'$ implies that the following conditions are satisfied:

(Inv) for all $p\in\BP$, $s\in V(p)$ iff $s'\in V'(p)$;

($\Box$-Forth) if $sRt$ for some $t$, then there is a $t'$ such that $s'Rt'$ and $tZt'$;

($\Box$-Back) if $s'Rt'$ for some $t'$, then there is a $t$ such that $sRt$ and $tZt'$.

\medskip

We say that $(\M,s)$ and $(\M',s')$ are {\em $\Box$-bisimilar}, notation: $(\M,s)\Kbis(\M',s')$, if there exists a $\Box$-bisimulation $Z$ between $\M$ and $\M'$ such that $sZs'$. When the models involved are clear, we write it $s\Kbis s'$ for brevity.
\end{definition}

The following result will be used in Proposition \ref{prop.kbisvskwbis}.

\begin{proposition}\label{prop.Kbisimilar} Let $\M=\lr{S,R,V}$ and $\M'=\lr{S',R',V'}$ be two models, and $s\in S$ and $s'\in S'$. Then
$(\M,s)\Kbis(\M',s')$ implies the following conditions:
\begin{enumerate}
\item \label{prop.bisimilar-one} For all $p\in\BP$, $s\in V(p)$ iff $s'\in V'(p)$;
\item \label{prop.bisimilar-two} If $sRt$, then there is a $t'$ in $\M'$ such that $s'R't'$ and $t\Kbis t'$;
\item \label{prop.bisimilar-three} If $s'R't'$, then there is a $t$ in $\M$ such that $sRt$
and $t\Kbis t'$.
\end{enumerate}
\end{proposition}

\begin{proof}
Follows directly from the fact that $\Kbis$ is a $\Box$-bisimulation and the definition of $\Box$-bisimulation.
\end{proof}

However, the notion of $\Box$-bisimulation is too refined for the logic $\LEA$, as will be shown below. The following  example arises in the proof of Proposition \ref{prop.lessexp}:

$$
\xymatrix{\mathcal{M}:\ \ \ s:p\ar@(ur,ul) & & & \mathcal{N}:\ \ \ t:p}
$$

It is not hard to show that $\M$ and $\N$ are both image-finite models, and that $(\M,s)$ and $(\N,t)$ satisfy the same $\LEA$-formulas, but they are not $\Box$-bisimilar. Therefore, we need to redefine a suitable notion of bisimulation for $\LEA$.

\begin{definition}[$\circ$-bisimulation]\label{def.circ-bis}
Let $\M=\langle S,R,V\rangle$. A nonempty binary relation $Z$ over $S$ is called a {\em $\circ$-bisimulation} on $\M$, if $sZs'$ implies that the following conditions are satisfied:

(Inv) for all $p\in\BP$, $s\in V(p)$ iff $s'\in V(p)$;

($\circ$-Forth) if $sRt$ and $(s,t)\notin Z$ for some $t$, then there is a $t'$ such that $s'Rt'$ and $tZt'$;

($\circ$-Back) if $s'Rt'$ and $(s',t')\notin Z$ for some $t'$, then there is a $t$ such that $sRt$ and $tZt'$.

\medskip

We say that $(\M,s)$ and $(\M',s')$ are {\em $\circ$-bisimilar}, notation: $(\M,s)\kwbis(\M',s')$, if there exists a $\circ$-bisimulation $Z$ on the disjoint union of $\M$ and $\M'$ such that $sZs'$.
\end{definition}

The following proposition states that we can build more sophisticated $\circ$-bisimulations from the simpler ones. For the proof details we refer to Appendix \ref{appendix.proofs}.

\begin{proposition}\label{prop.bis-union}
If $Z$ and $Z'$ are both $\circ$-bisimulations on $\M$, then $Z\cup Z'$ is also a $\circ$-bisimulation on $\M$.
\end{proposition}

In particular, by Definition~\ref{def.circ-bis}, one can see that $\circ$-bisimilarity is the largest $\circ$-bisimulation. And also, $\circ$-bisimilarity is an equivalence relation. Note that the proof is highly nontrivial. For the proof details we refer to Appendix \ref{appendix.proofs}.

\begin{proposition}\label{prop.max-bis}
The $\circ$-bisimilarity $\kwbis$ is an equivalence relation.
\end{proposition}

The following proposition indicates the relationship between $\circ$-bisimilarity and $\Box$-bisimilarity: $\circ$-bisimilarity is strictly weaker than $\Box$-bisimilarity. This corresponds to the fact that $\LEA$ is strictly weaker than $\ML$.

\begin{proposition}\label{prop.kbisvskwbis}
Let $(\M,s)$ and $(\M',s')$ be pointed models. If $(\M,s)\Kbis (\M',s')$, then $(\M,s)\kwbis(\M',s')$; but the converse does not hold.
\end{proposition}

\begin{proof}
Suppose that $(\M,s)\Kbis (\M',s')$. Define $Z=\{(x,x')\mid x\Kbis x',~x\in \M,~x'\in \M'\}$. We will show that $Z$ is a $\circ$-bisimulation on the disjoint union of $\M$ and $\M'$ with $sZs'$.

First, by supposition, we have $sZs'$, thus $Z$ is nonempty. We need only check that $Z$ satisfies the three conditions of $\circ$-bisimulation. Assume that $xZx'$. By definition of $Z$, we obtain $x\Kbis x'$. Using item~\ref{prop.bisimilar-one} of Proposition~\ref{prop.Kbisimilar}, we have $x$ and $x'$ satisfy the same propositional variables, thus (Inv) holds. For ($\circ$-Forth), suppose that $xRy$ and $(x,y)\notin Z$ for some $y$, then using item~\ref{prop.bisimilar-two} of Proposition~\ref{prop.Kbisimilar}, we get there exists $y'$ in $\M'$ such that $x'R'y'$ and $y\Kbis y'$, thus $yZy'$. The proof for condition ($\circ$-Back) is similar, by using item \ref{prop.bisimilar-three} of Proposition~\ref{prop.Kbisimilar}.

For the converse, recall the example in Proposition~\ref{prop.lessexp}. There, let $Z=\{(s,s),(s,t)\}$. It is not hard to show that $Z$ is a $\circ$-bisimulation on the disjoint union of $\M$ and $\N$, thus $(\M,s)\kwbis(\N,t)$. However, $(\M,s)\not\Kbis(\N,t)$, as $\M,s\nvDash\Box\bot$ but $\N,t\vDash\Box\bot$.
\end{proof}

The following result says that $\LEA$-formulas are invariant under $\circ$-bisimilarity. This means that $\LEA$-formulas cannot distinguish $\circ$-bisimilar models.
\begin{proposition}\label{prop.invariance}
Let $(\M,s)$ and $(\M',s')$ be pointed models. If $(\M,s)\kwbis(\M',s')$, then $(\M,s)\equkw(\M',s')$. In other words, $\circ$-bisimilarity implies $\circ$-equivalence.
\end{proposition}

\begin{proof}
Assume that $(\M,s)\kwbis(\M',s')$, then there is a $\circ$-bisimulation $Z$ on the disjoint union of $\M$ and $\M'$ such that $sZs'$. We need to show that for any $\phi\in \LEA$, $\M,s\vDash\phi$ iff $\M',s'\vDash\phi$.

The proof continues by induction on the structure of $\phi$. The non-trivial case is $\circ\phi$.

Suppose that $\M,s\nvDash\circ\phi$. Then $s\vDash\phi$ but there exists $t$ such that $sRt$ and $t\nvDash\phi$. By the induction hypothesis, $(s,t)\notin Z$. Then by ($\circ$-Forth) that there exists $t'$ such that $s'R't'$ and $tZt'$, thus $(\M,t)\kwbis (\M',t')$. From $s\kwbis s'$ and the induction hypothesis and $s\vDash\phi$, it follows that $s'\vDash\phi$. Analogously, we can infer $t'\nvDash\phi$. Therefore $\M',s'\nvDash\circ\phi$. For the converse use ($\circ$-Back).
\end{proof}

With the notion of $\circ$-bisimulation, we can simplify the proofs in the previous sections. We here take Proposition~\ref{prop.lessexp-d} as an example, to show that $(\M,s)$ and $(\N,s')$ therein are $\circ$-bisimilar, rather than using simultaneous induction. For this, we define $Z=\{(s,s'),(t,t'),(t,t)\}$\footnote{Note that in order to guarantee $Z$ is indeed a $\circ$-bisimulation, the pair $(t,t)$ must be contained in $Z$.}. We can show that $Z$ is indeed a $\circ$-bisimulation on the disjoint union of $\M$ and $\N$, thus $(\M,s)\kwbis (\N,s')$.\weg{ and $t\kwbis t'$. By Prop.~\ref{prop.invariance}, we have for all $\phi\in\LEA$, (i) $\M,s\vDash\phi$ iff $\N,s'\vDash\phi$, and (ii) $\M,t\vDash\phi$ iff $\N,t'\vDash\phi$.}

\medskip

For the converse, we have

\begin{proposition}[Hennessy-Milner Theorem]\label{prop.hennessy-milner-theorem}
Let $\M$ and $\M'$ be both image-finite models and $s\in\M$ and $s'\in\M'$. Then $(\M,s)\equkw(\M',s')$ iff $(\M,s)\kwbis(\M',s')$.
\end{proposition}

\begin{proof}
Let $\M$ and $\M'$ be both image-finite models and $s\in\M$ and $s'\in\M'$. Based on Proposition~\ref{prop.invariance}, we need only to show the direction from left to right. Assume that $(\M,s)\equkw(\M',s')$, we need to show that $\equkw$ is a $\circ$-bisimulation on the disjoint union of $\M$ and $\M'$, which implies $(\M,s)\kwbis(\M',s')$. It suffices to show the condition ($\circ$-Forth), as the proof for ($\circ$-Back) is similar.

Suppose that there exists $t$ such that $sRt$ and $s\not\equkw t$, to show for some $t'$ it holds that $s'R't'$ and $t\equkw t'$. Since $s\not\equkw t$, there is a $\phi\in\LEA$ such that $s\vDash\phi$ but $t\nvDash\phi$, and thus $s\nvDash\circ\phi$ due to $sRt$. By assumption, we have $s'\vDash\phi$ and $s'\nvDash\circ\phi$, and thus there exists $v'$ such that $s'R'v'$ and $v'\nvDash\phi$. Let $S'=\{t'\in\M'\mid s'R't'\}$. It is easy to see that $S'\neq \emptyset$. As $\M'$ is image-finite, $S'$ must be finite, say $S'=\{t_1',t_2',\cdots,t_n'\}$. If there is no $t'_i\in S'$ such that $t\equkw t'_i$, then for every $t_i'\in S'$ there exists $\phi_i\in\LEA$ such that $t\vDash\phi_i$ but $t_i'\nvDash\phi_i$. It follows that $t\vDash\phi_1\land\cdots\land\phi_n$, and thus $t\nvDash\phi_1\land\cdots\land\phi_n\to\phi$; furthermore, from $s\vDash\phi$ follows that $s\vDash\phi_1\land\cdots\land\phi_n\to\phi$. Hence $s\nvDash\circ(\phi_1\land\cdots\land\phi_n\to\phi)$. Note that for all $t_i'\in S'$, $t_i'\nvDash\phi_1\land\cdots\land\phi_n$, thus $t_i'\vDash\phi_1\land\cdots\land\phi_n\to\phi$. We also have $s'\vDash\phi_1\land\cdots\land\phi_n\to\phi$, and then $s'\vDash\circ(\phi_1\land\cdots\land\phi_n\to\phi)$, which is contrary to the assumption and $s\nvDash\circ(\phi_1\land\cdots\land\phi_n\to\phi)$.  Therefore, we have for some $t'$ it holds that $s'R't'$ and $t\equkw t'$.
\end{proof}

If we remove the condition of `image-finite', then $\kwbis$ does not coincide with $\equkw$.

\begin{example}\label{example.not-m-saturated}
Consider two models $\M=\lr{S,R,V}$ and $\M'=\lr{S',R',V'}$, where $S=\mathbb{N}\cup\{s\}$, $R=\{(s,n)\mid n\in\mathbb{N}\},V(p_n)=\{n\}$
and $S'=\mathbb{N}\cup\{s',\omega\}$, $R'=\{(s',n)\mid n\in\mathbb{N}\}\cup\{(s',\omega)\}$, and $V'(p_n)=\{n\}$. This can be visualized as follows:
$$
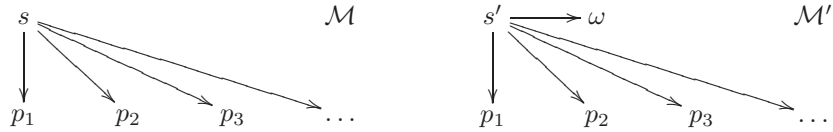
\xymatrix{
s \ar[d]\ar[dr]\ar[drr]\ar[drrr]&&&\M\\
p_1&p_2&p_3& \dots
}
\qquad
\qquad
\xymatrix{
s' \ar[d]\ar[dr]\ar[drr]\ar[drrr]\ar[r]&\omega&&\M'\\
p_1&p_2&p_3& \dots
}
$$
We have:
\begin{itemize}
\item Neither of $\M$ and $\M'$ is image-finite, as $s$ and $s'$ both have infinite many successors.
\item $(\M,s)\equkw(\M',s')$. By induction on $\phi\in\LEA$, we show that for any $\phi$, $\M,s\vDash\phi$ iff $\M',s'\vDash\phi$. The non-trivial case is $\circ\phi$, that is to show, $\M,s\vDash\circ\phi$ iff $\M',s'\vDash\circ\phi$. The direction from right to left follows directly from $R(s)\subseteq R'(s')$. For the other direction, suppose that $\M,s\vDash\circ\phi$. Then $s\vDash\phi$ implies for any $n\in\mathbb{N}$, $n\vDash\phi$. By the induction hypothesis, $s'\vDash\phi$ implies for any $n\in\mathbb{N}$, $n\vDash\phi$. As $\phi$ is finite, it contains only finitely many propositional variables. Without loss of generality, we may assume that $n$ is the largest number of subscripts of propositional variables occurring in $\phi$. Then by induction on $\phi$, we can show that $n+1\vDash\phi$ iff $\omega\vDash\phi$. Thus $s'\vDash\phi$ implies for any $n\in\mathbb{N}\cup\{\omega\}$, $n\vDash\phi$. Therefore $\M',s'\vDash\circ\phi$, as desired.
\item $(\M,s)\not\kwbis(\M',s')$. Suppose, for a contradiction, that $(\M,s)\kwbis(\M',s')$, then there exists a $\circ$-bisimulation $Z$ such that $sZs'$. Now we have $s'R'\omega$. And also $(s',\omega)\notin Z$, for otherwise $s'\kwbis \omega$, thus e.g. $s'\vDash\circ\neg p_1$ iff $\omega\vDash\circ\neg p_1$, contrary to the fact that $s'\nvDash\circ\neg p_1$ but $\omega\vDash\circ\neg p_1$. By the condition ($\circ$-Back), we obtain that there exists $m\in\mathbb{N}$ such that $sRm$ and $mZ\omega$, thus $m\kwbis \omega$. However, $m\vDash p_m$ but $\omega\nvDash p_m$, contradiction.
\end{itemize}
\end{example}

Proposition~\ref{prop.hennessy-milner-theorem} can be extended to the following stronger proposition. Here by {\em $\LEA$-saturated model} we mean, given any $s$ in this model and any set $\Gamma\subseteq\LEA$, if all of finite subsets of $\Gamma$ are satisfiable in the successors of $s$, then $\Gamma$ is also satisfiable in the successors of $s$.

\begin{proposition}\label{prop.delta-saturated}
Let $\M$ and $\M'$ be both $\LEA$-saturated models and $s\in\M$ and $s'\in\M'$. Then $(\M,s)\equkw(\M',s')$ iff $(\M,s)\kwbis (\M',s')$.
\end{proposition}

\begin{proof}
Based on Proposition \ref{prop.invariance}, we need only show the direction from left to right.

Let $\M=\lr{S,R,V}$ and $\M'=\lr{S',R',V'}$ be both $\LEA$-saturated models. Suppose that $(\M,s)\equkw(\M',s')$, we will show that $\equkw$ is a $\circ$-bisimulation on the disjoint union of $\M$ and $\M'$, which implies $(\M,s)\kwbis (\M',s')$. It suffices to show the condition ($\circ$-Forth) holds, as the proof of ($\circ$-Back) is similar.

Assume that $sRt$ and $s\not\equkw t$ for some $t$, to show there exists $t'$ such that $s'R't'$ and $t\equkw t'$. Let $\Gamma=\{\phi\in\LEA\mid t\vDash\phi\}$. It is clear that $t\vDash\Gamma$. Then for any finite $\Sigma\subseteq \Gamma$, $t\vDash\bigwedge\Sigma$. As $s\not\equkw t$, there exists $\psi\in\LEA$ such that $s\vDash\psi$ but $t\nvDash\psi$, thus $s\vDash\bigwedge\Sigma\to\psi$ but $t\nvDash\bigwedge\Sigma\to\psi$, hence $s\nvDash\circ(\bigwedge\Sigma\to\psi)$. If for any $u'$ with $s'R'u'$ we have $u'\nvDash\bigwedge\Sigma$, then $u'\vDash\bigwedge\Sigma\to\psi$. Since $s\equkw s'$ and $s\vDash\psi$, it follows that $s'\vDash\psi$, thus $s'\vDash\bigwedge\Sigma\to\psi$, hence $s'\vDash\circ(\bigwedge\Sigma\to\psi)$, contrary to $s\equkw s'$ and $s\nvDash\circ(\bigwedge\Sigma\to\psi)$. Therefore there exists $u'$ such that $s'R'u'$ and $u'\vDash\bigwedge\Sigma$. Because $\M'$ is $\LEA$-saturated, for some $t'$ we have $s'R't'$ and $t'\vDash\Gamma$. Furthermore, $t\equkw t'$: given any $\phi\in\LEA$, if $t\vDash\phi$, then $\phi\in\Gamma$, hence $t'\vDash\phi$; if $t\nvDash\phi$, i.e., $t\vDash\neg\phi$, then $\neg\phi\in\Gamma$, hence $t'\vDash\neg\phi$, i.e., $t'\nvDash\phi$, as desired.
\end{proof}

The condition `$\LEA$-saturated' is also indispensable, which can also be illustrated with Example \ref{example.not-m-saturated}. In that example, $\M$ is not $\LEA$-saturated. To see this point, note that the set $\{\neg p_1,\neg p_2,\cdots,\neg p_n\}$ is finitely satisfiable in the successors of $s$, but the set itself is {\em not} satisfiable in the successors of $s$. In the meantime, $(\M,s)\equkw (\M',s')$ but $(\M,s)\not\kwbis(\M',s')$.

\medskip

We have seen from Definition~\ref{def.circ-bis} that the notion of $\circ$-bisimulation is quite different from that of $\Box$-bisimulation. However, it is surprising that the notion of $\circ$-bisimulation contraction is very similar to that of $\Box$-bisimulation contraction, by simply replacing $\Kbis$ with $\kwbis$.
\begin{definition}[$\circ$-bisimulation contraction] Let $\M=\lr{S,R,V}$ be a model. The {\em $\circ$-bisimulation contraction} of $\M$ is the quotient structure $[\M]=\lr{[S],[R],[V]}$ such that
\begin{itemize}
\item $[S]=\{[s]\mid s\in S\}$, where $[s]=\{t\in S\mid s\kwbis t\}$;
\item $[s][R][t]$ iff there exist $s'\in[s]$ and $t'\in[t]$ such that $s'Rt'$\weg{ and $s'\not\kwbis t'$};
\item $[V](p)=\{[s]\mid s\in V(p)\}$ for all $p\in \BP$.
\end{itemize}
\end{definition}

Under this definition, we obtain that the contracted model (via $\kwbis$) is $\circ$-bisimilar to the original one, and that the $\mathcal{S}5$-model property is preserved under $\circ$-bisimulation contraction. For the proof details we refer to Appendix \ref{appendix.proofs}.

\begin{proposition}\label{prop.circ-bis-con}
Let $\M=\lr{S,R,V}$ be a model, and let $[\M]=\lr{[S],[R],[V]}$ be the $\circ$-bisimulation contraction of $\M$. Then for any $s\in S$, we have $([\M],[s])\kwbis(\M,s)$.
\end{proposition}

\begin{proposition}\label{prop.s5-preserved}
Let $\M=\lr{S,R,V}$ be a model, and let $[\M]=\lr{[S],[R],[V]}$ be the $\circ$-bisimulation contraction of $\M$. If $\M$ is an $\mathcal{S}5$-model, then $[\M]$ is also an $\mathcal{S}5$-model.
\end{proposition}

\section{Characterization Results}\label{sec.char}

As $\vDash\circ\phi\leftrightarrow(\phi\to\Box\phi)$, the logic of essence and accident can be seen as a fragment of standard modal logic, and thus also a fragment of first-order logic. In this section we characterize the logic of essence and accident within standard modal logic and within first-order logic. To make our exposition self-contained, we introduce some definitions and results from e.g. \cite{blackburnetal:2001} without proofs, refer to Appendix \ref{appendix.preliminaries}.

Since $\LEA$ can be viewed as a fragment of $\ML$, every $\LEA$-formula can be seen as an $\ML$-formula. By Proposition~\ref{prop.equk}, we have

\begin{lemma}\label{lem.equkw}
Let $\M$ be a model and $s\in\M$. Then $ue(\M)$ is $\LEA$-saturated and $(\M,s)\equkw(ue(\M),\pi_s)$.
\end{lemma}

From Lemma \ref{lem.equkw} and Proposition~\ref{prop.delta-saturated}, it follows that
\begin{lemma}\label{lem.ue}
Let $(\M,s)$ and $(\N,t)$ be pointed models. Then $(\M,s)\equkw(\N,t)$ implies $(ue(\M),\pi_s)\kwbis(ue(\N),\pi_t)$.
\end{lemma}

We are now close to prove two characterization results: the logic of essence and accident is the $\circ$-bisimulation-invariant fragment of standard modal logic and of first-order logic. In the following, by an $\ML$-formula $\phi$ (resp. a first-order formula $\alpha$) {\em is invariant under $\circ$-bisimulation}, we mean for any models $(\M,s)$ and $(\N,t)$, if $(\M,s)\kwbis(\N,t)$, then $\M,s\vDash\phi$ iff $\N,t\vDash\phi$ (resp. $\M,s\vDash \alpha$ iff $\N,t\vDash\alpha$).
\begin{theorem}\label{thm.delta-characterization}
An $\ML$-formula is equivalent to an $\LEA$-formula iff it is invariant under $\circ$-bisimulation.
\end{theorem}

\begin{proof}
Based on Proposition~\ref{prop.invariance}, we need only show that the direction from right to left. For this, suppose that an $\ML$-formula $\phi$ is invariant under $\circ$-bisimulation.

Let $MOC(\phi)=\{t(\psi)\mid \psi\in\LEA,\phi\vDash t(\psi)\}$, where $t$ is a translation function which recursively translates every $\LEA$-formula into the corresponding $\ML$-formulas; in particular, $t(\circ\psi)=t(\psi)\to\Box t(\psi)$.

If we can show that $MOC(\phi)\vDash\phi$, then by Compactness Theorem of modal logic, there exists a finite set $\Gamma\subseteq MOC(\phi)$ such that $\bigwedge\Gamma\vDash \phi$, i.e., $\vDash\bigwedge\Gamma\to \phi$. Besides, the definition of $MOC(\phi)$ implies that $\phi\vDash\bigwedge\Gamma$, i.e., $\vDash \phi\to\bigwedge\Gamma$, and thus $\vDash\bigwedge\Gamma\lra\phi$. Since every $\gamma\in\Gamma$ is a translation of an $\LEA$-formula, so is $\Gamma$. Then we are done.

Assume that $\M,s\vDash MOC(\phi)$, to show that $\M,s\vDash\phi$. Let $\Sigma=\{t(\psi)\mid \psi\in\LEA,\M,s\vDash t(\psi)\}$. We now claim $\Sigma\cup\{\phi\}$ is satisfiable: otherwise, by Compactness Theorem of modal logic again, there exists finite $\Sigma'\subseteq \Sigma$ such that $\phi\vDash\neg\bigwedge\Sigma'$, thus $\neg\bigwedge\Sigma'\in MOC(\phi)$. By assumption, we obtain $\M,s\vDash\neg\bigwedge\Sigma'$. However, the definition of $\Sigma$ and $\Sigma'\subseteq \Sigma$ implies $\M,s\vDash\bigwedge\Sigma'$, contradiction.

Thus we may assume that $\N,t\vDash\Sigma\cup\{\phi\}$. We can show $(\M,s)\equkw(\N,t)$ as follows: for any $\psi\in\LEA$, if $\M,s\vDash\psi$, then $\M,s\vDash t(\psi)$, and then $t(\psi)\in\Sigma$, thus $\N,t\vDash t(\psi)$, hence $\N,t\vDash\psi$; if $\M,s\nvDash\psi$, i.e., $\M,s\vDash\neg\psi$, then $\M,s\vDash t(\neg\psi)$, and then $t(\neg\psi)\in\Sigma$, thus $\N,t\vDash t(\neg\psi)$, hence $\N,t\vDash\neg\psi$, i.e. $\N,t\nvDash\psi$.

We now construct the ultrafilter extensions of $\M$ and $\N$, denoted by $ue(\M)$ and $ue(\N)$, respectively. According to the fact that $(\M,s)\equkw(\N,t)$ and Lemma \ref{lem.ue}, we have $(ue(\M),\pi_s)\kwbis(ue(\N),\pi_t)$. Since $\N,t\vDash\phi$, by Lemma \ref{lem.equkw}, we have $ue(\N),\pi_t\vDash\phi$. From supposition it follows that $ue(\M),\pi_s\vDash\phi$. Using Lemma \ref{lem.equkw} again, we conclude that $\M,s\vDash\phi$.
\end{proof}

\begin{theorem}
A first-order formula is equivalent to an $\LEA$-formula iff it is invariant under $\circ$-bisimulation.
\end{theorem}

\begin{proof}
Based on Proposition~\ref{prop.invariance}, we need only show the direction from right to left. For this, suppose that a first-order formula $\alpha$ is invariant under $\circ$-bisimulation, then by Proposition~\ref{prop.kbisvskwbis}, we have that $\alpha$ is also invariant under $\Box$-bisimulation. From van Benthem Characterization Theorem (Proposition~\ref{thm.vanbenthem}), it follows that $\alpha$ is equivalent to an $\ML$-formula $\phi$. From this and supposition, it follows that $\phi$ is invariant under $\circ$-bisimulation. By Theorem~\ref{thm.delta-characterization}, $\phi$ is equivalent to an $\LEA$-formula. Therefore, $\alpha$ is equivalent to an $\LEA$-formula.
\end{proof}

\section{Axiomatizations}\label{sec.axiomatization}\label{sec.axiomatizations}

This section deals with the axiomatization for the logic $\LEA$ over various classes of frames. We first handle the minimal system.

\begin{definition} [Axiomatic system \SLEA] The axiomatic system \SLEA\ consists of all propositional tautologies (\TAUT), uniform substitution (\SUB), modus ponens (\MP), plus the following axioms and inference rule:
\[
\begin{array}{ll}
\KwTop& \circ\top\\

\EquiKw&\neg p\to\circ p\\

\KwCon&\circ p\land\circ q\to\circ (p\land q)\\

\R&\text{From }\phi\to\psi \text{ infer }\circ\phi\land\phi\to\circ\psi\\
\end{array} \]

A {\em derivation} from $\Gamma$ to $\phi$ in $\SLEA$, notation: $\Gamma\vdash_{\SLEA}\phi$, is a finite sequence of $\LEA$-formulas in which each formula is either an instantiation of an axiom, or an element of $\Gamma$, or the result of applying an inference rule to prior formulas in the sequence. Formula $\phi$ is {\em provable} in $\SLEA$, or a {\em theorem}, notation: $\vdash\phi$, if there is a derivation from the empty set $\emptyset$ to $\phi$ in $\SLEA$.
\end{definition}

Intuitively, Axiom $\KwTop$ says that tautologies are not accidentally true (i.e. $\neg\bullet\top$); Axiom $\EquiKw$ says that whatever is accidentally true is always true (i.e. $\bullet p\to p$); Axiom $\KwCon$ says that if the conjunction is accidentally true, then at least one conjunct thereof is accidentally true (i.e. $\bullet(p\land q)\to\bullet p\vee\bullet q$); Rule $\R$ stipulates the {\em almost} monotonicity of the essence operator.

When it comes to completeness, any of the axioms $\KwTop,\EquiKw$ and $\KwCon$ is indispensable in the system $\SLEA$, otherwise the subsystems will be {\em not} complete. As for Axiom $\KwTop$, define a nonstandard semantics $\Vdash$ as $\vDash$, except that all formulas of the form $\circ\phi$ are interpreted as $\neg\phi$. We can check under this semantics, $\SLEA-\KwTop$ is sound, but $\KwTop$ is not valid, which means that $\KwTop$ is not provable in $\SLEA-\KwTop$. However, $\KwTop$ is valid under the standard semantics $\vDash$. Therefore, $\SLEA-\KwTop$ is {\em  not} complete with respect to the semantics $\vDash$. As for Axiom $\EquiKw$, define another nonstandard semantics $\Vvdash$ as $\vDash$, except that all formulas of the form $\circ\phi$ are interpreted as $\phi$. One can show that under the semantics $\Vvdash$, the subsystem $\SLEA-\EquiKw$ is sound, but $\EquiKw$ is {\em not} valid. Thus there is a validity (i.e. Axiom $\EquiKw$) under the standard semantics $\vDash$, which is unprovable in $\SLEA-\EquiKw$, and hence $\SLEA-\EquiKw$ is {\em not} complete with respect to the semantics $\vDash$.

As to the indispensability of Axiom $\KwCon$, the situation is more complicated. For this, we need to switch the interpretations of $\circ$ and $\bullet$, and the soundness of a system is defined as ``all of the theorems involved in the occurrence of $\circ$ are invalid'', where the notion of validity is defined as usual (see Definition~\ref{def.semantics}). Then one can check that under this specification, the subsystem $\SLEA-\KwCon$ is sound, but Axiom $\KwCon$ is valid, thus $\KwCon$ is {\em not} provable in $\SLEA-\KwCon$. However, $\KwCon$ is valid under the semantics $\vDash$, hence $\SLEA-\KwCon$ is {\em not} complete with respect to the semantics $\vDash$.

From the indispensability of the axioms $\KwTop,\EquiKw$ and $\KwCon$, we have also shown that all of the three axioms are {\em independent} in the system $\SLEA$.

Note that our axiomatic system $\SLEA$ is equivalent to, but slightly different from Steinsvold's $B_K$ in \cite{Steinsvold:2008}. We can show easily that $\SLEA$ is sound with respect to the class of all frames. \weg{It is not hard to show that $\SLEA$ contains the same theorems as $B_K$.}

Using Axiom $\KwCon$ and Rule $\SUB$, we can show by induction on $n\in\mathbb{N}$ that
\begin{proposition}\label{prop.circ-conj}
$\vdash\circ\phi_1\land\cdots\land\circ\phi_n\to\circ(\phi_1\land\cdots\land\phi_n)$.
\end{proposition}

We are now ready to build the canonical model for $\SLEA$.

\weg{\begin{lemma}[Truth Lemma]\label{lem.truthlem} Let $s\in S^c$ and $\phi\in\LEA$. We have
$$\M^c,s\vDash\phi\Longleftrightarrow\phi\in s.$$
\end{lemma}

\begin{proof}
See Appendix~\ref{appendix.proofs} for proof details.
\end{proof}

Based on Lemma \ref{lem.truthlem}, it is a standard exercise to show that
\begin{theorem}[Completeness of $\SLEA$]\label{thm.comp-k}
$\SLEA$ is sound and strongly complete with respect to the class of all frames.
\end{theorem}

Note that the canonical relation $R^c$ is reflexive, thus $R^c$ is also serial.
\begin{theorem}[Completeness of $\SLEA$ over $\mathcal{T}$-frames]
$\SLEA$ is sound and strongly complete with respect to the class of $\mathcal{T}$-frames.
\end{theorem}

\begin{theorem}[Completeness of $\SLEA$ over $\mathcal{D}$-frames]
$\SLEA$ is sound and strongly complete with respect to the class of $\mathcal{D}$-frames.
\end{theorem}}

\begin{definition}[Canonical model for $\SLEA$]\label{def.cm-k} The model $\M^c=\lr{S^c,R^c,V^c}$ is the {\em canonical model} of $\SLEA$, where
\begin{itemize}
\item $S^c=\{s\mid s\text{ is a maximal consistent set for }\SLEA\}$;
\item For any $s,t\in S^c$, $sR^ct$ iff (for all $\phi\in\LEA$, if $\circ\phi\land\phi\in s$, then $\phi\in t$) and $s\neq t$;
\item $V^c(p)=\{s\in S^c\mid p\in s\}$.
\end{itemize}
\end{definition}

The canonical model here is not reflexive, which is consistent with the semantics of $\circ$, in contrast to the definition in \cite{Steinsvold:2008} (see Section~\ref{sec.comparison}).

\begin{lemma}[Truth Lemma]\label{lem.truthlem} Let $s\in S^c$ and $\phi\in\LEA$. We have
$$\M^c,s\vDash\phi\Longleftrightarrow\phi\in s.$$
\end{lemma}

\begin{proof}
By induction on $\phi$. The only nontrivial case is $\circ\phi$, that is to show, $\M^c,s\vDash\circ\phi\Longleftrightarrow\circ\phi\in s$.

`$\Longleftarrow$': Suppose towards contradiction that $\circ\phi\in s$ but $\M^c,s\nvDash\circ\phi$, then $s\vDash\phi$ but there is a $t\in S^c$ with $sR^ct$ and $t\nvDash\phi$. By the induction hypothesis, we have $\phi\in s$ but $\phi\notin t$. Thus $\circ\phi\land\phi\in s$. Since $sR^ct$, we obtain $\phi\in t$, contradiction.

`$\Longrightarrow$': Suppose $\circ\phi\notin s$, to show $\M^c,s\nvDash\circ\phi$. By the induction hypothesis, we need only show that $\phi\in s$ but there is a $t\in S^c$ with $sR^ct$ and $\neg\phi\in t$. First, $\phi\in s$ follows from the supposition $\circ\phi\notin s$, Axiom $\EquiKw$ and Rule $\SUB$. Besides, we show that the set $\{\psi\mid \circ\psi\land\psi\in s\}\cup\{\neg\phi\}$ is consistent.

The proof proceeds as follows: if the set is not consistent, then there exist $\psi_1,\cdots,\psi_n\in\{\psi\mid \circ\psi\land\psi\in s\}$\footnote{Note that Axiom $\KwTop$ provides the nonempty of the set $\{\psi\mid \circ\psi\land\psi\in s\}$.} such that $\vdash\psi_1\land\cdots\land\psi_n\to\phi$. Using Rule $\R$, we get $\vdash\circ(\psi_1\land\cdots\land\psi_n)\land(\psi_1\land\cdots\land\psi_n)\to\circ\phi$. From this and Proposition~\ref{prop.circ-conj} follows that $\vdash\circ\psi_1\land\cdots\land\circ\psi_n\land(\psi_1\land\cdots\land\psi_n)\to\circ\phi$. Since $\circ\psi_i\land\psi_i\in s$ for all $i\in[1,n]$, we have $\circ\phi\in s$, contrary to the supposition.

We have thus shown that $\{\psi\mid \circ\psi\land\psi\in s\}\cup\{\neg\phi\}$ is consistent. By Lindenbaum's Lemma, there is a $t\in S^c$ such that $\{\psi\mid \circ\psi\land\psi\in s\}\cup\{\neg\phi\}\subseteq t$. Since $\phi\in s$ but $\phi\notin t$, we obtain $s\neq t$. Thus $sR^ct$ and $\neg\phi\in t$, as desired.
\end{proof}

Based on Lemma \ref{lem.truthlem}, it is a standard exercise to show that
\begin{theorem}[Completeness of $\SLEA$ over $\mathcal{K}$-frames]\label{thm.comp-k}
$\SLEA$ is sound and strongly complete with respect to the class of all frames.
\end{theorem}

The same story goes with $\SLEA$ and the class of serial frames. But note that $\M^c$ is not necessarily serial. Thus we need to transform $\M^c$ into a serial model, in the meanwhile the truth-values of $\LEA$-formulas should be preserved.

\begin{theorem}[Completeness of $\SLEA$ over $\mathcal{D}$-frames]\label{thm.comp-d}
$\SLEA$ is sound and strongly complete with respect to the class of $\mathcal{D}$-frames.
\end{theorem}

\begin{proof}
Define $\M^{\bf D}=\lr{S^c,R^{\bf D},V^c}$, where $S^c$ and $V^c$ is the same as in Definition~\ref{def.cm-k}, and $R^{\bf D}=R^c\cup\{(t,t)\mid t\text{~is an endpoint in~}\M^c\}$.\footnote{The method, called `reflexivizing the endpoints', is also used in \cite[Theorem~5.6]{Fanetal:2015}.} Now it is obvious that $\M^{\bf D}$ is serial.

\weg{We also need}It suffices to show that the truth-values of $\LEA$-formulas are preserved under the model transformation. That is to show: for all $s\in S^c$, for all $\phi\in\LEA$, we have $\M^c,s\vDash\phi$ iff $\M^{\bf D},s\vDash\phi$. The nontrivial case is $\circ\phi$. If $s$ is an endpoint in $\M^c$, then by semantics, $\M^c,s\vDash\circ\phi$ and $\M^{\bf D},s\vDash\circ\phi$, thus we have $\M^c,s\vDash\circ\phi$ iff $\M^{\bf D},s\vDash\circ\phi$. If $s$ is {\em not} an endpoint in $\M^c$, then the claim is clear.
\weg{Now suppose $\Gamma\nvdash\phi$. Then $\Gamma\cup\{\neg\phi\}$ is $\SLEA$-consistent. Using Lindenbaum's Lemma, there is an $s\in \M^c$ such that $\Gamma\cup\{\neg\phi\}\subseteq s$. By Truth Lemma (Lemma~\ref{lem.truthlem}), $\M^c,s\vDash\Gamma\cup\{\neg\phi\}$. From the above $(\star)$, it follows that $\M^{\bf D},s\vDash\Gamma\cup\{\neg\phi\}$. Moreover, $\M^{\bf D}$ is serial. Therefore, $\Gamma\nvDash_{F_{\bf D}}\phi$, where $F_{\bf D}$ is the class of ${\bf D}$-frames.}
\end{proof}

The same story also goes with $\SLEA$ and the class of reflexive frames. However, according to the definition of $R^c$, $\M^c$ is {\em not} reflexive, thus we need to transform $\M^c$ into a reflexive model. Notice that the truth-values of $\LEA$-formulas should be preserved under the model transformation.

\begin{theorem}[Completeness of $\SLEA$ over $\mathcal{T}$-frames]\label{thm.comp-t}
$\SLEA$ is sound and strongly complete with respect to the class of $\mathcal{T}$-frames.
\end{theorem}

\begin{proof}
Define $\M^c$ as in Definition~\ref{def.cm-k}. Define $\M^{\bf T}=\lr{S^c,R^{\bf T},V^c}$ as $\M^c$, except that $R^{\bf T}$ is the reflexive closure of $R^c$, i.e. $R^{\bf T}=R^c\cup\{(s,s)\mid s\in S^c\}$, equivalently, $sR^{\bf T}t$ iff (for all $\phi\in\LEA$, if $\circ\phi\land\phi\in s$, then $\phi\in t$) or $s=t$. It is now obvious that $R^{\bf T}$ is reflexive.

It suffices to show that the truth-values of $\LEA$-formulas are preserved under the model transformation. That is to show, for all $s\in S^c$, for all $\phi\in\LEA$, we have $\M^c,s\vDash\phi$ iff $\M^{\bf T},s\vDash\phi$. The proof proceeds with induction on $\phi$. The nontrivial case is $\circ\phi$, as follows.
\[\begin{array}{ll}
&\M^c,s\vDash\circ\phi\\
\Longleftrightarrow& \M^c,s\vDash\phi \text{~implies for all }t\in S^c,\text{ if }sR^ct,\text{ then }\M^c,t\vDash\phi\\
\Longleftrightarrow& \M^c,s\vDash\phi\text{ implies }\M^c,s\vDash\phi\text{ and for all }t\in S^c,\text{ if }sR^ct,\text{ then }\M^c,t\vDash\phi\\
\stackrel{\text{IH}}\Longleftrightarrow& \M^{\bf T},s\vDash\phi\text{ implies }\M^{\bf T},s\vDash\phi\text{ and for all }t\in S^c,\text{ if }sR^ct,\text{ then }\M^{\bf T},t\vDash\phi\\
\stackrel{sR^{\bf T}s}\Longleftrightarrow& \M^{\bf T},s\vDash\phi\text{ implies for all }t\in S^c,\text{ if }sR^{\bf T}t,\text{ then }\M^{\bf T},t\vDash\phi \\
\Longleftrightarrow& \M^{\bf T},s\vDash\circ\phi.\\
\end{array}\]
\weg{We only need to prove that the truth lemma holds for $\SLEA$ under the definition of $\M^{\bf T}$. That is to show, for all $\phi\in \LEA$,
$$\M^{\bf T},s\vDash\phi\Longleftrightarrow\phi\in s.$$

The nontrivial case is $\circ\phi$, i.e., to show $\M^{\bf T},s\vDash\circ\phi\Longleftrightarrow\circ\phi\in s$. The direction from left to right is similar to `$\Longrightarrow$' in Lemma~\ref{lem.truthlem}.

It suffices to show the other direction. For this, suppose, for a contradiction, that $\circ\phi\in s$ but $\M^{\bf T},s\nvDash\circ\phi$, then $s\vDash\phi$ but there is a $t\in S^c$ such that $sR^{\bf T}t$ and $t\nvDash\phi$, and then by induction hypothesis, $\phi\in s$ and $\phi\notin t$, thus $s\neq t$. From this and $sR^{\bf T}t$ and $\circ\phi\land\phi\in s$, we obtain $\phi\in t$, contradiction.}
\end{proof}

\medskip

We now consider the extensions of the system $\SLEA$. The table below indicates the extra axioms and the corresponding systems, with on the right-hand side the classes of frames for which we will demonstrate completeness.
\[
\begin{array}{|l|l|l|l|}
\hline
\text{Notation}&\text{Axioms}&\text{Systems}&\text{Frame classes}\\
\hline
\KwTr& \circ p\land p\to\circ\circ p& \SLCLTr=\SLCL+\KwTr&4~(\mathcal{S}4)\\
\KwB& p\to\circ(\circ\neg p\to p)&\SLCLB=\SLCL+\KwB&\mathcal{B} ~(\mathcal{TB})\\
\KwEuc&\neg \circ\neg p\to\circ (\circ\neg p\to p)&\SLCLBEuc=\SLCLB+\KwEuc&\mathcal{B}5~(\mathcal{S}5)\\
\hline
\end{array}
\]

Note that Axiom $\KwTr$ is different from the axiom B4 in~\cite[p.~95]{Steinsvold:2008}, i.e. $\circ p\land p\to\circ(\circ p\land p)$. One can show that B4 is provable in $\SLCLTr$, with the aid of Axioms $\KwTr$, $\KwCon$ and Rule $\SUB$. It is shown in~\cite[Prop.~3.6]{Steinsvold:2008} that $\SLCL+\text{B4}$ is sound and complete with respect to the class of $4$-frames (weakly transitive frames, and also $\mathcal{S}4$-frames). The same argument goes with the system $\SLCLTr$. But we will show that $\SLCLTr$ is sound and strongly complete with respect to the class of $4$-frames within our framework. And we can see that $\SLCLTr$ is simpler than $\SLCL+\text{B4}$.

We can compute the above axioms from the standard ones in modal logic. But note that the point here is Proposition~\ref{prop.equivalent}. In other words, the underlying class of frames is $F_{\mathcal{T}}$, the class of reflexive frames, rather than the class of all frames.
\begin{align}
&~\Box p\to\Box\Box p\label{41}\\
\Lra&~\circ p\land p\to\circ(\circ p\land p)\land (\circ p\land p)\label{42}\\
\Lra&~\circ p\land p\to\circ(\circ p\land p)\label{43}
\end{align}
The equivalent transition from (\ref{41}) to (\ref{42}) follows from Proposition~\ref{prop.equivalent}. By simplification, we obtain the axiom B4 in~\cite[p.~95]{Steinsvold:2008}, i.e. (\ref{43}).
\begin{align}
&~p\to\Box\Diamond p\label{b1}\\
\Lra&~p\to\Box(p\lor\neg\circ\neg p)\label{b2}\\
\Lra&~p\to(p\lor\neg\circ\neg p)\land\circ(p\lor\neg\circ\neg p)\label{b3}\\
\Lra&~p\to\circ(\circ\neg p\to p)\label{b4}
\end{align}
The equivalent transitions from (\ref{b1}) to (\ref{b3}) follow from Proposition~\ref{prop.equivalent}. By simplification and using Rule $\REKw$ (i.e. $\dfrac{\phi\lra\psi}{\circ\phi\lra\circ\psi}$), we obtain Axiom $\circ{\bf B}$, i.e. (\ref{b4}).
\weg{\begin{align}
&~\Diamond p\to\Box\Diamond p\label{51}\\
\Leftrightarrow& ~p\lor\Diamond p\to\Box(p\lor\Diamond p)\land(p\lor\Diamond p)\label{52}\\
\Lra&~ (\neg p\to\Diamond p)\to\Box(\neg p\to\Diamond p)\land(\neg p\to\Diamond p)\label{53}\\
\Lra&~ (\neg p\to\neg \circ\neg p)\to\circ (\neg p\to\neg \circ\neg p)\land (\neg p\to\neg \circ\neg p)\label{54}\\
\Lra&~ (p\lor \neg \circ\neg p)\to\circ (\neg p\to\neg \circ\neg p)\label{55}\\
\Lra&~ (p\lor \neg \circ\neg p)\to\circ (\circ\neg p\to p)\label{56}\\
\Lra&~(p\to\circ (\circ\neg p\to p))\land(\neg \circ\neg p\to\circ (\circ\neg p\to p))\label{57}\\
\Rightarrow&~\neg \circ\neg p\to\circ (\circ\neg p\to p)\label{58}
\end{align}
The equivalent transition from (\ref{51}) to (\ref{52}) is due to the validity of $\Diamond \phi\lra \phi\lor\Diamond \phi$ and of $\Box\phi\lra\Box\phi\land\phi$ on reflexive frames. The equivalent transition from (\ref{53}) to (\ref{54}) follows from Prop.~\ref{prop.equivalent}}
\begin{align}
&~\Diamond p\to\Box\Diamond p\label{51}\\
\Lra&~ (p\lor \neg \circ\neg p)\to\circ (p\lor \neg \circ\neg p)\land (p\lor \neg \circ\neg p)\label{55}\\
\Lra&~ (p\lor \neg \circ\neg p)\to\circ (\circ\neg p\to p)\label{56}\\
\Lra&~(p\to\circ (\circ\neg p\to p))\land(\neg \circ\neg p\to\circ (\circ\neg p\to p))\label{57}\\
\Rightarrow&~\neg \circ\neg p\to\circ (\circ\neg p\to p)\label{58}
\end{align}
The equivalent transition from (\ref{51}) to (\ref{55})  follows from Proposition~\ref{prop.equivalent}. The equivalent transition from (\ref{56}) to (\ref{57}) is due to the validity of $(\phi\vee\psi\to\chi)\lra(\phi\to\chi)\land(\psi\to\chi)$. The {\em implicative} (rather than equivalent) transition from (\ref{57}) to (\ref{58}) is because Axiom $\circ{\bf B}$ is invalid on Euclidean frames. In this way, we get the axiom (\ref{58}), i.e. $\circ{\bf5}$, or following the term in~\cite[p.~100]{Steinsvold:2008}, B5.

From the above transition from (\ref{51}) to (\ref{58}), we can see that the standard axiom ${\bf 5}$ in modal logic is equivalent to $\circ{\bf B}~\&\circ{\bf5}$, rather than $\circ{\bf5}$. This tells us that $\circ{\bf B}~\&\circ{\bf5}$ is the desired axiom for characterizing the logic of essence and accident over symmetric and Euclidean frames, but $\circ{\bf5}$ may {\em not} be the desired axiom for characterizing this logic over Euclidean frames, as one can show.\footnote{In \cite[page 101]{Steinsvold:2008}, the author claimed without a proof that $B_K+B5$ is the logic of $K5_{EA}$, which means that $B_K+B5$, equivalently, our $\SLCL+\KwEuc$, is sound and {\em complete} with respect to the class of Euclidean frames. However, by using his canonical model (see also Section~\ref{sec.comparison}), we cannot see how the canonical relation therein is provided to be Euclidean.\label{footnote}}

\medskip

When considering the case for transitivity, the difficulty will arise. Because it is possible that $sR^ct$ and $tR^cu$ and $u=s$. By the definition of $R^c$, we do not have $sR^cu$. We call this kind of world $s$ (viz. $u$) a {\em non-transitive world w.r.t. $R^c$}.\weg{ Otherwise, we call it a transitive world w.r.t. $R^c$.} Thus we need to transform $\M^c$ into a transitive model. Notice that the truth-values of $\LEA$-formulas should be preserved under the transformation.
\begin{theorem}[Completeness of $\SLCLTr$ over $4$-frames]\label{thm.comp-4}
$\SLCLTr$ is sound and strongly complete with respect to the class of $4$-frames.
\end{theorem}

\weg{\begin{proof}
Suppose $sR^ct$ and $tR^cu$, to show $sR^cu$. Consider two cases.
\begin{itemize}
\item $s\neq u$.

\item $s=u$. Define $\M^{\bf Tr}=\lr{S^c,R^{\bf Tr},V^c}$, where $S^c$ is the set of maximal $\SLCLTr$-consistent sets, $V^c$ is the same as in Definition \ref{def.cm-k}, and $R^{\bf Tr}=R^c\cup\{(s,s)\mid sR^ct, tR^cs\}$, where $R^c$ is defined as in Definition \ref{def.cm-k}.
\end{itemize}
\end{proof}}

\begin{proof}
Define $\M^c$ as in Definition~\ref{def.cm-k} w.r.t. $\SLCLTr$. Define $\M^{\bf Tr}=\lr{S^c,R^{\bf Tr},V^c}$ as $\M^c$, except that $R^{\bf Tr}=R^c\cup\{(s,s)\mid sR^ct,tR^cs \text{ for some }t\in S^c\}$.

We first show that $R^{\bf Tr}$ is transitive. Assume for any $s,t,u\in S^c$ that $sR^{\bf Tr}t$ and $tR^{\bf Tr}u$, to show $sR^{\bf Tr}u$. We consider the following cases.
\begin{itemize}
\item $sR^ct$ and $tR^cu$. If $s=u$, then by definition of $R^{\bf Tr}$, it is obvious that $sR^{\bf Tr}u$. We only need to consider $s\neq u$. In this case, suppose for any $\phi\in\LEA$ that $\circ\phi\land\phi\in s$, by Axiom $\KwTr$ and Rule $\SUB$, we get $\circ\circ\phi\in s$. Then from $sR^ct$, we can infer $\circ\phi\land\phi\in t$. Combining this and $tR^cu$, we obtain $\phi\in u$. Therefore, $sR^cu$, and thus $sR^{\bf Tr}u$.
\item $sR^ct$, and for some $s'$, $tR^cs'$ and $s'R^cu$ and $t=u$. Then it is obvious that $sR^cu$, thus $sR^{\bf Tr}u$.
\item for some $s'\in S^c$,  $sR^cs'$ and $s'R^ct$ and $s=t$, and $tR^cu$. Then it is obvious that $sR^cu$, thus $sR^{\bf Tr}u$.
\item for some $s',s''\in S^c$,  $sR^cs'$ and $s'R^ct$ and $s=t$, and $tR^cs''$ and $s''R^cu$ and $t=u$. Then $sR^cs'$ and $s'R^cu$ and $s=u$ for some $s'\in S^c$. Thus $sR^{\bf Tr}u$.
\end{itemize}
Either case implies $sR^{\bf Tr}u$.

It suffices to show that the truth-values of $\LEA$-formulas are preserved under the model transformation. That is to show, for any $s\in S^c$, for any $\phi\in\LEA$, we have:
$\M^c,s\vDash\phi\Longleftrightarrow\M^{\bf Tr},s\vDash\phi.$ The proof proceeds with induction on $\phi$. The nontrivial case is $\circ\phi$. If $s$ is {\em not} a non-transitive world w.r.t. $R^c$, then the claim is clear. Otherwise, i.e., if $s$ is a non-transitive world w.r.t. $R^c$, then
\[\begin{array}{ll}
&\M^c,s\vDash\circ\phi\\
\Longleftrightarrow&\M^c,s\vDash\phi\text{ implies for all }t\in S^c,\text{ if }sR^ct,\text{ then }\M^c,t\vDash\phi\\
\Longleftrightarrow&\M^c,s\vDash\phi\text{ implies }\M^c,s\vDash\phi\text{ and for all }t\in S^c,\text{ if }sR^ct,\text{ then }\M^c,t\vDash\phi\\
\stackrel{\text{IH}}\Longleftrightarrow&\M^{\bf Tr},s\vDash\phi\text{ implies }\M^{\bf Tr},s\vDash\phi\text{ and for all }t\in S^c,\text{ if }sR^ct,\text{ then }\M^{\bf Tr},t\vDash\phi\\
\stackrel{sR^{\bf Tr}s}\Longleftrightarrow&\M^{\bf Tr},s\vDash\phi\text{ implies for all }t\in S^c,\text{ if }sR^{\bf Tr}t,\text{ then }\M^{\bf Tr},t\vDash\phi \\
\Longleftrightarrow& \M^{\bf Tr},s\vDash\circ\phi.\\
\end{array}\]

According to the previous analysis, this completes the proof.
\end{proof}

\begin{theorem}[Completeness of $\SLCLTr$ over $\mathcal{S}4$-frames] \label{thm.comp-s4}
$\SLCLTr$ is sound and strongly complete with respect to the class of $\mathcal{S}4$-frames.
\end{theorem}

\begin{proof}
Define $\M^{\bf T}$ as in the proof of Theorem~\ref{thm.comp-t} w.r.t.~$\SLCLTr$. By Theorem~\ref{thm.comp-t}, we only need to show that $R^{\bf T}$ is transitive.

Suppose for any $s,t,u\in S^c$ that $sR^{\bf T}t$ and $tR^{\bf T}u$, to show $sR^{\bf T}u$. If $s=u$, then it is clear that $sR^{\bf T}u$. We only need to consider the case where $s\neq u$. In this case, according to the definition of $R^{\bf T}$, we consider the following subcases (it is impossible that $s=t$ and $t=u$, because in this case we already have $s\neq u$).
\begin{itemize}
\item $sR^ct$ and $tR^cu$. Since $s\neq u$, by a similar argument to the corresponding part of the first item in Theorem~\ref{thm.comp-4}, we can obtain $sR^cu$.
\item $sR^ct$ and $t=u$. It is obvious that $sR^cu$.
\item $s=t$ and $tR^cu$. It is obvious that $sR^cu$.
\end{itemize}
Either case implies $sR^{\bf T}u$, which completes the proof.
\end{proof}

For the completeness of $\SLCLB$, according to Theorem~\ref{thm.comp-k}, we only need to show that
\begin{proposition}\label{prop.sym}
$R^c$ is symmetric.
\end{proposition}

\begin{proof}
Let $s,t\in S^c$. Suppose $sR^ct$ (thus $s\neq t$), to show $tR^cs$. For this, assume for any $\phi$ such that $\circ\phi\land\phi\in t$, we need only show $\phi\in s$. If $\phi\notin s$, i.e. $\neg\phi\in s$, then using Axiom $\circ{\bf B}$ and Rule $\SUB$, we obtain $\circ(\circ\neg\neg\phi\to\neg\phi)\in s$, viz. $\circ(\circ\phi\to\neg\phi)\in s$. Besides, from $\neg\phi\in s$ it follows that $\circ\phi\to\neg\phi\in s$. We have thus shown that $\circ(\circ\phi\to\neg\phi)\land(\circ\phi\to\neg\phi)\in s$. By supposition and the definition of $R^c$, we conclude that $\circ\phi\to\neg\phi\in t$, that is, $\circ\phi\land\phi\notin t$, in contradiction to the assumption. Therefore $\phi\in s$.
\end{proof}

From Proposition~\ref{prop.definable-b}, Theorem~\ref{thm.comp-k} and Proposition~\ref{prop.sym}, it is immediate that

\begin{theorem}[Completeness of $\SLCLB$ over $\mathcal{B}$-frames]\label{thm.comp-kb}
$\mathbf{KB^\circ}$ is sound and strongly complete with respect to the class of $\mathcal{B}$-frames.
\end{theorem}

\begin{theorem}[Completeness of $\SLCLB$ over $\mathcal{TB}$-frames]\label{thm.comp-tb}
$\SLCLB$ is sound and strongly complete with respect to the class of $\mathcal{TB}$-frames.
\end{theorem}

\begin{proof}
Define $\M^{\bf T}$ as in the proof of Theorem~\ref{thm.comp-t} w.r.t. $\SLCLB$. By Theorem~\ref{thm.comp-t}, we need only show that $R^{\bf T}$ is symmetric.

Let $s,t\in S^c$. Suppose $sR^{\bf T}t$, to show $tR^{\bf T}s$. By supposition, $sR^ct$ or $s=t$. If $sR^ct$, then by a similar argument to Proposition~\ref{prop.sym}, we can show that $tR^cs$; if $s=t$, then $t=s$. Either case implies $tR^{\bf T}s$, as desired.
\end{proof}

As we mentioned in the footnote~\ref{footnote}, system $\SLCL+\KwEuc$ may not be complete with respect to the class of Euclidean frames. Despite this, we indeed have the following result, i.e. Theorem~\ref{thm.comp-b5}. For this, however, we cannot provide that the canonical relation $R^c$ in Definition~\ref{def.cm-k} is Euclidean, since it may be the case that $sR^ct$ and $sR^cu$ but $t=u$ (thus it does not hold that $tR^cu$). We call this kind of world a {\em non-Euclidean world w.r.t. $R^c$}. We need to handle these special worlds, to transform $\M^c$ into an Euclidean model. Note that the transformation need to keep the symmetry of $R^c$, and also preserve the truth-values of $\LEA$-formulas.

\begin{theorem}[Completeness of $\SLCLBEuc$ over $\mathcal{B}5$-frames]\label{thm.comp-b5}
$\SLCLBEuc$ is sound and strongly complete with respect to the class of $\mathcal{B}5$-frames.
\end{theorem}

\begin{proof}
The validity of Axiom $\KwEuc$ can be derived from a similar argument to Proposition~\ref{prop.new}.\ref{prop.newtwo}, and the soundness of $\SLCLB$ follows from Theorem~\ref{thm.comp-kb}.

Define $\M^c$ as in Definition~\ref{def.cm-k} w.r.t. $\SLCLBEuc$. Define $\M^{\bf Euc}=\lr{S^c,R^{\bf Euc},V^c}$ as $\M^c$, except that $R^{\bf Euc}=R^c\cup\{(t,t)\mid sR^ct\text{ for some }s\in S^c\}$.

We first show $R^{\bf Euc}$ is Euclidean. Assume for any $s,t,u\in S^c$ that $sR^{\bf Euc}t$ and $sR^{\bf Euc}u$, to show that $tR^{\bf Euc}u$. According to the definition of $R^{\bf Euc}$, we consider the following cases.

\begin{itemize}
\item $sR^ct$ and $sR^cu$. Consider two subcases.
    \begin{itemize}
    \item $t=u$. It is obvious that $tR^{\bf Euc}u$.

    \item $t\neq u$. Suppose for any $\phi$ that $\circ\phi\land\phi\in t$, we need to show $\phi\in u$, from which we have $tR^cu$. Since $R^c$ is symmetric (Proposition~\ref{prop.sym}), from $sR^ct$ it follows that $tR^cs$. Then the supposition implies $\phi\in s$. Moreover, we have $\circ\phi\in s$, for otherwise, $\neg\circ\phi\in s$, by Axiom $\KwEuc$ and Rule $\SUB$, $\circ(\circ\phi\to\neg\phi)\in s$; we also have $\circ\phi\to\neg\phi\in s$, then from $sR^ct$ follows $\circ\phi\to\neg\phi\in t$, contrary to the supposition. We have thus shown $\phi\land\circ\phi\in s$. This entails $\phi\in u$ due to $sR^cu$. Then $tR^cu$, and thus $tR^{\bf Euc}u$.
    \end{itemize}
\item $sR^ct$, and $s'R^cs$ and $s'R^cu$ and $s=u$ for some $s'\in S^c$. Then $uR^ct$. Since $R^c$ is symmetric (Proposition~\ref{prop.sym}), we obtain $tR^cu$, thus $tR^{\bf Euc}u$.
\item $sR^cu$, and $s'R^cs$ and $s'R^ct$ and $s=t$ for some $s'\in S^c$. Then it is clear that $tR^cu$, thus $tR^{\bf Euc}u$.
\item $s'R^cs$ and $s'R^ct$ and $s=t$ and $s''R^cs$ and $s''R^cu$ and $s=u$ for some $s',s''\in S^c$. Then $s'R^ct, s'R^cu$ and $t=u$, thus $tR^{\bf Euc}u$.
\end{itemize}

We then show $R^{\bf Euc}$ is symmetric. Suppose for any $s,t\in S^c$ that $sR^{\bf Euc}t$, to show $tR^{\bf Euc}s$. By supposition, we have either $sR^ct$, or $s'R^cs,s'R^ct,s=t$ for some $s'\in S^c$. If $sR^ct$, by Proposition~\ref{prop.sym}, we have $tR^cs$; if $s'R^cs,s'R^ct,s=t$ for some $s'\in S^c$, then $s'R^ct,s'R^cs,t=s$. Either case implies $tR^{\bf Euc}s$.

It suffices to show: for any $s\in S^c$, for any $\phi\in \LEA$, we have $\M^c,s\vDash\phi\Longleftrightarrow\M^{\bf Euc},s\vDash\phi$. That is to say, the truth-values of $\LEA$-formulas are preserved under the transformation. The proof proceeds with induction on $\phi$. The nontrivial case is $\circ\phi$.\weg{ If $s$ is {\em not} a non-Euclidean world w.r.t. $R^c$, then the claim is clear.} If $s$ is a non-Euclidean world w.r.t. $R^c$, then by semantics, $\M^c,s\vDash\circ\phi$ is equivalent to ($\M^c,s\vDash\phi$ implies for all $t\in S^c$, if $sR^ct$, then $\M^c,t\vDash\phi$), which is equivalent to ($\M^c,s\vDash\phi$ implies $\M^c,s\vDash\phi$ and for all $t\in S^c$, if $sR^ct$, then $\M^c,t\vDash\phi$). By induction hypothesis, this is equivalent to ($\M^{\bf Euc},s\vDash\phi$ implies $\M^{\bf Euc},s\vDash\phi$ and for all $t\in S^c$, if $sR^{c}t$, then $\M^{\bf Euc}, t\vDash\phi$). Since $sR^{\bf Euc}s$, this is equivalent to ($\M^{\bf Euc},s\vDash\phi$ implies for all $t\in S^c$, if $sR^{\bf Euc}t$, then $\M^{\bf Euc}, t\vDash\phi$), which means exactly $\M^{\bf Euc},s\vDash\circ\phi$. Otherwise, i.e., if $s$ is {\em not} a non-Euclidean world w.r.t. $R^c$, then it is obvious that $\M^c,s\vDash\circ\phi\Longleftrightarrow\M^{\bf Euc},s\vDash\circ\phi$.
\end{proof}

\begin{theorem}[Completeness of $\SLCLBEuc$ over $\mathcal{S}5$-frames]\label{thm.comp-s5}
$\SLCLBEuc$ is sound and strongly complete with respect to the class of $\mathcal{S}5$-frames.
\end{theorem}

\begin{proof}
Define $\M^{\bf T}$ as in the proof of Theorem~\ref{thm.comp-t} w.r.t. $\SLCLBEuc$. By Theorem~\ref{thm.comp-tb}, we need only show that $R^{\bf T}$ is Euclidean.

Suppose for any $s,t,u\in S^c$ such that $sR^{\bf T}t$ and $sR^{\bf T}u$, to show $tR^{\bf T}u$. If $s=t$, obviously $tR^{\bf T}u$. If $s=u$, then $uR^{\bf T}t$. Since $R^{\bf T}$ is symmetric (Theorem~\ref{thm.comp-tb}), we have $tR^{\bf T}u$. If $t=u$, by the definition of $R^{\bf T}$, we also have $tR^{\bf T}u$. So we only need to consider the case where $s\neq t$ and $s\neq u$ and $t\neq u$. Then $sR^ct$ and $sR^cu$. Analogous to the corresponding part of the first item in Theorem~\ref{thm.comp-b5}, we can infer $tR^cu$, thus $tR^{\bf T}u$, as desired.
\weg{Assume, for a contradiction, that there exists $\phi$ such that $\circ\phi\land\phi\in t$ but $\phi\notin u$. Since $R^c$ is symmetric (Proposition \ref{prop.sym}) and $sR^ct$, we have $tR^cs$, then $\phi\in s$ due to $\circ\phi\land\phi\in t$. We also have $\neg\circ\phi\in s$, for otherwise $\circ\phi\in s$, then $\circ\phi\land\phi\in s$, from which and $sR^cu$ we have $\phi\in u$, contrary to $\phi\notin u$. From $\neg\circ\phi\in s$, using Axiom $\KwEuc$ and Rule $\SUB$, we obtain $\circ(\circ\phi\to\neg\phi)\in s$. Besides, we can show $\circ\phi\to\neg\phi\in s$, thus $\circ(\circ\phi\to\neg\phi)\land(\circ\phi\to\neg\phi)\in s$. Thanks to $sR^ct$, we get $\circ\phi\to\neg\phi\in t$, contrary to the assumption. Therefore, $tR^cu$, thus $tR^{\bf T}u$. This completes the proof.}
\end{proof}

\weg{We also have
\begin{proposition}
For any $s,t,u\in S^c$ with $tR^cs$. If $sR^ct$ and $sR^cu$, then $tR^cu$.
\end{proposition}}

\section{Comparison with the literature}\label{sec.comparison}

Various axiomatizations for the logic $\LEA$ have been proposed in the literature. Inspired by \cite{DBLP:journals/ndjfl/Kuhn95}, a function $D$ is defined in \cite{Marcos:2005}, as $D(s)=\{\phi\mid\circ(\phi\vee\psi)\in s \text{ for every }\psi\}$ and then $sR^ct$ holds just in case $D(s)\subseteq t$. The function is simplified as $D(s)=\{\phi\mid \circ\phi\land\phi\in s\}$ in \cite{Steinsvold:2008}. In \cite[Prop.~4.1]{Kushida:2010}, a proof system $GS_0$ is established and shown to be complete, by proving $GS_0$ is equivalent to $K_{EA}$ in \cite{Marcos:2005}.

The canonical relation in \cite{Marcos:2005} may not apply to some other frame classes, just as the canonical relation defined in \cite{DBLP:journals/ndjfl/Kuhn95} in the setting of noncontingency logic (cf.~\cite[p.~101]{Fanetal:2015}). We here compare our method to that proposed in \cite{Steinsvold:2008} more detailedly, where, to show the completeness of $B_K$ (an equivalent axiomatization of $\SLEA$) and its extensions, the canonical model $\M^c=\lr{S^c,R^c,V^c}$ is defined as follows:
\begin{itemize}
\item $S^c=\{s\mid s\text{ is a maximal consistent set for }\SLEA\}$;
\item For any $s,t\in S^c$, $sR^ct$ iff for all $\phi\in\LEA$, if $\circ\phi\land\phi\in s$, then $\phi\in t$;
\item $V^c(p)=\{s\in S^c\mid p\in s\}$.
\end{itemize}

Under this definition, the completeness proof is simpler than ours. However, this definition has its defect: on one hand, the canonical relation, thus the canonical frame, is automatically provided to be reflexive; on the other hand, the semantics of $\LEA$ is defined on arbitrary frames, rather than on reflexive frames, and it is also shown in \cite[Prop.~3.5]{Steinsvold:2008} that $B_K$, equivalently, our $\SLEA$, is sound and complete with respect to the class of all frames. This means that there is a non-correspondence between syntax and semantics in the logic of essence and accident.

In comparison, by defining the canonical relation $R^c$ as in Definition~\ref{def.cm-k}, we do avoid the defect existing in \cite{Steinsvold:2008}, since our $\M^c$ is {\em not} reflexive.

Besides, our system $\SLCLTr$ is simpler than $\SLCL+\text{B4}$. And moreover, by using a translation from $\ML$ to $\LEA$ on {\em reflexive} frames, we obtain a complete axiomatization for symmetric frames, which, to our knowledge, is missing in the literature. We also have studied the model theory of $\LEA$, including the expressive power, frame definability, the proposed suitable notions of bisimulation and bisimulation contraction.

\section{Closing words}\label{sec.conclusion}

In this paper, we compared the relative expressivity of the logic of essence and accident $\LEA$ and modal logic, and study the frame definability of $\LEA$. We proposed a notion of bisimulation for $\LEA$, based on which we characterized this logic within modal logic and within first-order logic. We axiomatized the logic of essence and accident over various classes of frames, with a more suitable method than those in the literature. We found a method to compute certain axioms used to axiomatize this logic over special frames in the literature. As a side effect, we answered some open questions raised in \cite{Marcos:2005}.

As we claimed before, $\KwEuc$ may not be the desired axiom for characterizing $\LEA$ over Euclidean frames. We suspect that the validity $\KwEuc'$, i.e. $\neg p\to\circ(\circ\neg p \to p)$, on Euclidean frames is {\em not} provable in $\SLCL+\KwEuc$. Besides, We conjecture that $\SLCL+\KwEuc'$ is sound and strongly complete with respect to the class of Euclidean frames. We leave this for future work.
\bibliographystyle{plain}
\bibliography{biblio2014}

\appendix

\section{Omitted proofs}\label{appendix.proofs}

\jieproof{prop.new}
For item \ref{prop.newone}, suppose $\mathcal{F}\vDash\forall x\forall y\forall z(xRy\land yRz\land x\neq y\land y\neq z\land x\neq z\to xRz)$, to show $\mathcal{F}\vDash\circ p\land p\to\circ(\circ p\land p)$. Assume, for a contradiction, that $\mathcal{F}\nvDash\circ p\land p\to\circ(\circ p\land p)$, then there exists $\M=\lr{S,R,V}$ based on $\mathcal{F}$ and $s\in S$ such that $\M,s\vDash\circ p\land p$ but $s\nvDash\circ(\circ p\land p)$. It follows that there is a $t\in S$ such that $sRt$ and $t\nvDash\circ p\land p$, thus $s\neq t$. Since $s\vDash\circ p\land p$, we have $t\vDash p$, and then $t\nvDash\circ p$, hence there is a $u\in S$ such that $tRu$ and $u\nvDash p$, and furthermore $t\neq u$ and $s\neq u$. Now by supposition, we obtain $sRu$. However, from $s\vDash\circ p\land p$ and $sRu$, we get $u\vDash p$, contradiction.

Now suppose $\mathcal{F}\nvDash\forall x\forall y\forall z(xRy\land yRz\land x\neq y\land y\neq z\land x\neq z\to xRz)$, to show $\mathcal{F}\nvDash\circ p\land p\to\circ(\circ p\land p)$. By supposition, there are $s,t,u$ such that $sRt,tRu,s\neq t,t\neq u, s\neq u$ but {\em not} $sRu$. Define a valuation $V$ on $\mathcal{F}=\lr{S,R}$ such that $V(p)=\{s\}\cup\{s'\in S\mid sRs'\}$. By definition, $\lr{F,V},s\vDash p$ and for all $s'$ such that $sRs'$ we have $s'\vDash p$, thus $s\vDash \circ p\land p$. Since $sRt$, we have $t\vDash p$. Since $s\neq u$ and not $sRu$, we obtain $u\nvDash p.$ From $tRu$ and $t\vDash p$ but $u\nvDash p$, it follows that $t\nvDash\circ p$, and then $t\nvDash \circ p\land p$. From this, $sRt$ and $s\vDash \circ p\land p$, we get $s\nvDash\circ(\circ p\land p)$. Therefore $s\nvDash\circ p\land p\to\circ(\circ p\land p)$, and we can now conclude that $\mathcal{F}\nvDash\circ p\land p\to\circ(\circ p\land p)$.

\medskip

For item \ref{prop.newtwo}, suppose $\mathcal{F}\vDash\forall x\forall y\forall z(xRy\land xRz\land x\neq y\land x\neq z\land y\neq z\to yRz)$, to show $\mathcal{F}\vDash\neg\circ\neg p\to\circ(\circ\neg p\to p)$. Assume, for a contradiction, that $\mathcal{F}\nvDash\neg\circ\neg p\to\circ(\circ\neg p\to p)$, then there exists $\M=\lr{S,R,V}$ based on $\mathcal{F}$ such that $\M,s\vDash\neg\circ\neg p$ but $s\nvDash\circ(\circ\neg p\to p)$. It follows that there is a $t\in S$ such that $sRt$ and $t\nvDash\circ\neg p\to p$, i.e. $t\vDash\circ \neg p\land \neg p$. Since $s\vDash\neg\circ\neg p$, we have $s\vDash\neg p$ and there is a $u\in S$ such that $sRu$ and $u\vDash p$. It is not hard to check that $s\neq t$, $s\neq u$ and $t\neq u$. Now by supposition, we obtain $tRu$. However, from $t\vDash\circ \neg p\land \neg p$ and $tRu$, we get $u\vDash\neg p$, i.e. $u\nvDash p$, contradiction.

Now suppose $\mathcal{F}\nvDash\forall x\forall y\forall z(xRy\land xRz\land x\neq y\land x\neq z\land y\neq z\to yRz)$, to show $\mathcal{F}\nvDash\neg\circ\neg p\to\circ(\circ\neg p\to p)$. By supposition, there are $s,t,u$ such that $sRt,sRu,s\neq t,s\neq u,t\neq u$ but {\em not} $tRu$. Define a valuation $V$ on $\mathcal{F}=\lr{S,R}$ such that $V(p)=\{u\}$. By definition and $s\neq u, t\neq u$, we have $s\nvDash p$ and $t\nvDash p$. Given any $t'$ such that $tRt'$, due to not $tRu$, we have $t'\neq u$, thus $t'\nvDash p$, and hence $t\vDash\circ\neg p\land \neg p$, i.e. $t\nvDash\circ \neg p\to p$. From $sRu, s\vDash \neg p, u\nvDash\neg p$, it follows that $s\nvDash\circ\neg p$, then $s\vDash\neg\circ\neg p$ and $s\vDash\circ\neg p\to p$. Since $sRt$, we can show that $s\nvDash\circ(\circ\neg p\to p)$. Therefore $s\nvDash\neg\circ\neg p\to\circ(\circ\neg p\to p)$, and we can now conclude that $\mathcal{F}\nvDash\neg\circ\neg p\to\circ(\circ\neg p\to p)$.

\medskip

\jieproof{prop.bis-union}
Suppose that $Z$ and $Z'$ are both $\circ$-bisimulations on $\M$, to show $Z\cup Z'$ is also a $\circ$-bisimulation on $\M$. Obviously, $Z\cup Z'$ is nonempty, since $Z,Z'$ are both non-empty. We need to check that $Z\cup Z'$ satisfies the three conditions of $\circ$-bisimulation. For this, assume that $(s,s')\in Z\cup Z'$. Then $sZs'$ or $sZ's'$.

(Inv):  If $sZs'$, then as $Z$ is a $\circ$-bisimulation, we have: given any $p\in\BP$, $s\in V(p)$ iff $s'\in V(p)$; if $sZ's'$, then as $Z'$ is a $\circ$-bisimulation, we also have: given any $p\in\BP$, $s\in V(p)$ iff $s'\in V(p)$. In both case we have that given any $p\in\BP$, $s\in V(p)$ iff $s'\in V(p)$.

($\circ$-Forth): Suppose that $sRt$ and $(s,t)\notin Z\cup Z'$, then $(s,t)\notin Z$ and $(s,t)\notin Z'$. If $sZs'$, then since $Z$ is a $\circ$-bisimulation on $\M$, there exists $t'$ such that $s'Rt'$ and $(t,t')\in Z$, and hence $(t,t')\in Z\cup Z'$; if $sZ's'$, then since $Z'$ is a $\circ$-bisimulation on $\M$, there exists $t'$ such that $s'Rt'$ and $(t,t')\in Z'$, and hence also $(t,t')\in Z\cup Z'$. Therefore in both cases, there exists $t'$ such that $s'Rt'$ and $(t,t')\in Z\cup Z'$.

($\circ$-Back): The proof is similar to that of ($\circ$-Forth).

\medskip

\jieproof{prop.max-bis}
We need only show that $\kwbis$ satisfies the three properties of an equivalence relation.

Reflexivity: Given any model $\M=\lr{S,R,V}$ and $s\in S$, to show that $(\M,s)\kwbis(\M,s)$. For this, define $Z=\{(w,w)\mid w\in S\}$. First, $Z$ is nonempty, as $sZs$. We need only show that $Z$ satisfies the three conditions of $\circ$-bisimulation. Suppose that $wZw$.

It is obvious that $w$ and $w$ satisfy the same propositional variables, thus (Inv) holds; suppose that $wRt$ and $(w,t)\notin Z$ for some $t\in S$, then obviously, there exists $t'=t$ such that $wRt'$ and $tZt'$, thus ($\circ$-Forth) holds; the proof of ($\circ$-Back) is analogous.

\medskip

Symmetry: Given any models $\M=\lr{S,R,V}$ and $\M'=\lr{S',R',V'}$ and $s\in S$ and $s'\in S'$, assume that $(\M,s)\kwbis (\M',s')$, to show that $(\M',s')\kwbis (\M,s)$. By assumption, we have that there exists $\circ$-bisimulation $Z$ with $sZs'$. Define $Z'=\{(w,w')\mid w\in S,~w'\in S',~ w'Zw\}\cup\{(w,t)\mid w,t\in S,~wZt\}\cup\{(w',t')\mid w',t'\in S',~w'Zt'\}.$ First, since $sZs'$, we have $(s',s)\in Z'$, thus $Z'$ is nonempty. We need only check that $Z'$ satisfies the three conditions of $\circ$-bisimulation. Suppose that $wZ'w'$.

By supposition, we have $w'Zw$. Using (Inv) of $Z$, we have that $w'$ and $w$ satisfy the same propositional variables, then of course $w$ and $w'$ satisfy the same propositional variables, thus (Inv) holds. For ($\circ$-Forth), suppose that $wRt$ and $(w,t)\notin Z'$ for some $t\in S$, then by definition of $Z'$, $(w,t)\notin Z$. Using ($\circ$-Back) of $Z$, we infer that there exists $t'\in S'$ such that $w'R't'$ and $t'Zt$, thus $tZ't'$. The proof of ($\circ$-Back) is similar, by using ($\circ$-Forth) of $Z$.

\medskip

Transitivity: Given any models $\M=\lr{S^{\M},R^{\M},V^{\M}}$,~$\N=\lr{S^{\N},R^{\N},V^{\N}}$,\\
$\mathcal{O}=\lr{S^{\mathcal{O}},R^{\mathcal{O}},V^{\mathcal{O}}}$ and $s\in S^{\M},~t\in S^{\N},~u\in S^{\mathcal{O}}$, assume that $(\M,s)\kwbis(\N,t)$ and $(\N,t)\kwbis(\mathcal{O},u)$, to show that $(\M,s)\kwbis(\mathcal{O},u)$. By assumption, we have that there exists $\circ$-bisimulation $Z_1$ on the disjoint union of $\M$ and $\N$ such that $sZ_1t$, and there exists $\circ$-bisimulation $Z_2$ on the disjoint union of $\N$ and $\mathcal{O}$ such that $tZ_2u$. We need to find a $\circ$-bisimulation $Z$ on the disjoint union of $\M$ and $\mathcal{O}$.

Define $Z=\{(x,z)\mid x\in S^{\M},~z\in S^{\mathcal{O}},\text{ there is a }y\in S^{\N}\text{ such that }xZ_1y,\\
yZ_2z\}\cup\{(x,x')\mid x,x'\in S^{\M},~xZ_1x'\}\cup\{(z,z')\mid z,z'\in S^{\mathcal{O}},~zZ_2z'\}\cup\{(x,x')\mid x,x'\in S^{\M},\text{ there are }y,y'\in S^{\N}\text{ such that }yZ_2y',~xZ_1y,~x'Z_1y'\}\cup\{(z,z')\mid z,z'\in S^{\mathcal{O}},\text{ there are }y,y'\in S^{\N}\text{ such that }yZ_1y',~yZ_2z,~y'Z_2z'\}.$ First, since $sZ_1t$ and $tZ_2u$, by the first part of the definition of $Z$, we have $sZu$, thus $Z$ is nonempty. We need only check that $Z$ satisfies the three conditions of $\circ$-bisimulation. Suppose that $xZz$. Then by the first part of the definition of $Z$, there is a $y\in S^{\N}$ such that $xZ_1y$ and $yZ_2z$.

(Inv): as $Z_1$ and $Z_2$ are both $\circ$-bisimulations, $x$ and $y$ satisfy the same propositional variables, and $y$ and $z$ satisfy the same propositional variables. Then $x$ and $z$ satisfy the same propositional variables.

($\circ$-Forth): suppose that $xR^{\M}x'$ and $(x,x')\notin Z$ for some $x'\in S^{\M}$, then by the second part of the definition of $Z$, we obtain $(x,x')\notin Z_1$. From this, $xZ_1y$ and ($\circ$-Forth) of $Z_1$, it follows that there exists $y'\in S^{\N}$ such that $yR^{\N}y'$ and $x'Z_1y'$. Using $xZ_1y,x'Z_1y',(x,x')\notin Z$ and the fourth part of the definition of $Z$, we get $(y,y')\notin Z_2$. From this, $yZ_2z$ and ($\circ$-Forth) of $Z_2$, it follows that there exists $z'\in S^{\mathcal{O}}$ such that $zR^{\mathcal{O}}z'$ and $y'Z_2z'$. Since $x'Z_1y'$ and $y'Z_2z'$, by the first part of the definition of $Z$, we obtain $x'Zz'$. We have shown that, there exists $z'\in S^{\mathcal{O}} $ such that $zR^{\mathcal{O}}z'$ and $x'Zz'$, as desired.

($\circ$-Back): the proof is similar to that of ($\circ$-Forth), but in this case we use the third and fifth parts of the definition of $Z$, rather than the second or fourth parts of the definition of $Z$.

\medskip

\jieproof{prop.circ-bis-con}
Define $Z=\{([w],w)\mid w\in S\}\weg{\cup\{(w,v)\mid w,v\in S, w\kwbis v\}}\cup\{([w],[v])\mid w,v\in S, w\kwbis v\}$. First, since $S$ is nonempty, $Z$ is nonempty. We need to show that $Z$ satisfies the three conditions of $\circ$-bisimulation, which entails $([\M],[s])\kwbis(\M,s)$. Assume that $[w]Zw$.

(Inv): by the definition of $[V]$.

($\circ$-Forth): suppose that $[w][R][v]$ and $([w],[v])\notin Z$, then by definition of $[R]$, there exist $w'\in[w]$ and $v'\in[v]$ such that $w'Rv'$. As $[w]=[w']$ and $[v]=[v']$, we get from the supposition that $([w'],[v'])\notin Z$. By definition of $Z$, we obtain $w'\not\kwbis v'$. Because $\kwbis$ is a $\circ$-bisimulation, from $w\kwbis w'$ it follows that there exists $u$ such that $wRu$ and $u\kwbis v'$, thus $[u]=[v']=[v]$. It is clear that $[u]Zu$, that is, $[v]Zu$.

($\circ$-Back): suppose that $wRv$ and $(w,v)\notin Z$. By definition of $[R]$, we have $[w][R][v]$. Obviously, $[v]Zv$.

\medskip

\jieproof{prop.s5-preserved}
Suppose that $\M=\lr{S,R,V}$ is an $\mathcal{S}5$-model, to show that $[\M]=\lr{[S],[R],[V]}$ is also an $\mathcal{S}5$-model. We need to show that $[R]$ is an equivalence relation, that is, $[R]$ satisfies the properties of reflexivity, symmetry and transitivity.
The nontrivial case is transitivity. For this, given any $[s],[t],[u]\in[S]$, assume that $[s][R][t]$ and $[t][R][u]$, we need only show that $[s][R][u]$.

By assumption and the definition of $[R]$, there exists $s'\in[s],t'\in[t]$ such that $s'Rt'$, and there exists $t''\in[t],u'\in[u]$ such that $t''Ru'$. We now consider two cases:
\begin{itemize}
\item $t''\kwbis u'$. In this case, we have $[t'']=[u]$, then $[t]=[u]$, thus $[s][R][u]$.
\item $t''\not\kwbis u'$. In this case, using $t'\kwbis t''$ (as $t'\in[t]$ and $t''\in[t]$) and the fact that $\kwbis$ is a $\circ$-bisimulation, we obtain that there exists $u''$ such that $t'Ru''$ and $u''\kwbis u'$, thus $u''\in[u']=[u]$. Moreover, From $s'Rt',t'Ru''$ and the transitivity of $R$, it follows that $s'Ru''$. We have thus shown that there exists $s'\in[s],u''\in[u]$ with $s'Ru''$, therefore $[s][R][u]$.
\end{itemize}
In both cases we have $[s][R][u]$, as desired.

\section{Preliminaries} \label{appendix.preliminaries}

\begin{definition}[Ultrafilter Extension]\label{def.ue} Let $\M=\lr{S,R,V}$ be a model. We say that $ue(\M)=\lr{Uf(S),R^{ue},V^{ue}}$ is the \emph{ultrafilter extension} of $\M$, if
\begin{itemize}
\item $Uf(S)=\{u\mid u \text{ is an ultrafilter over }S\}$, where an ultrafilter $u\subseteq \wp(S)$ satisfies the following properties:
      \begin{itemize}
      \item $S\in u$, $\emptyset\notin u$,
      \item $X,Y\in u$ implies $X\cap Y\in u$,
      \item $X\in u$ and $X\subseteq Z\subseteq S$ implies $Z\in u$
      \item For all $X\in\mathcal{P}(S)$, $X\in u$ iff $-X\notin u$ ($-X$ means the complement of $X$)
      \end{itemize}
\item For all $s,t\in Uf(S)$, $sR^{ue}t$ iff for all $X\subseteq S$, $X\in t$ implies $\lambda(X)\in s$, where $\lambda(X)=\{w\in S\mid\text{ there exists }v \text{ such that }wRv \text{ and }v\in X\}$
\item $V^{ue}(p)=\{u\in Uf(S)\mid V(p)\in u\}$.
\end{itemize}
\end{definition}

\begin{definition}[Principle Ultrafilter] Let $S$ be a nonempty set. Given any $s\in S$, the \emph{principle ultrafilter} $\pi_s$ generated by $s$ is defined as $\pi_s=\{X\subseteq S\mid s\in X\}$. It can be shown that every principle ultrafilter is an ultrafilter.
\end{definition}

\begin{proposition}\label{prop.equk}
Let $\M$ be a model and $ue(\M)$ be its ultrafilter extension. Then $ue(\M)$ is $\ML$-saturated and $(\M,s)\equk (ue(\M),\pi_s)$.
\end{proposition}

\begin{theorem}[van Benthem Characterization Theorem]\label{thm.vanbenthem}
A first-order formula is equivalent to an $\ML$-formula iff it is invariant under $\Box$-bisimulation.
\end{theorem}

\end{document}